\documentclass[12pt,oneside]{amsart}
     \usepackage{amssymb}
      \usepackage[utf8]{inputenc}
     \usepackage[frenchb,english]{babel}
      \usepackage[T1]{fontenc}
      \usepackage{amsmath}
      \usepackage{amsfonts}
      \usepackage{lmodern}
      \usepackage[left=2cm,right=2cm,top=2cm,bottom=2cm]{geometry}
      \usepackage{variations}
      \usepackage{graphicx} 
      \usepackage{float}	
      \usepackage{geometry} 
      \usepackage{fancyhdr} 
      \usepackage{longtable}
      \usepackage{listings}
      \usepackage{subfigure}
      \usepackage{hyperref} 
      \usepackage[usenames,dvipsnames]{color} 
      \usepackage{lineno}
      \usepackage{enumitem}

      \theoremstyle{plain}
      \newtheorem{theorem}{Theorem}[section]

      \newtheorem{Proposition}[theorem]{Proposition}
      
      \theoremstyle{definition}
      \newtheorem{definition}[theorem]{Definition}
      
      \theoremstyle{remark}

      \numberwithin{equation}{section}

      \makeatletter
      \def\@setcopyright{}
      \def\serieslogo@{}
      \makeatother

\begin{document}
 \author{Véronique Maume-Deschamps}
 \address{Université de Lyon, Université Lyon 1, France, Institut Camille Jordan UMR 5208}
 \email{veronique.maume@univ-lyon1.fr}
 
 
 \author{Didier Rulli\`ere}
 \address{Université de Lyon, Université Lyon 1, France, Laboratoire SAF EA 2429}
 \email{didier.rulliere@univ-lyon1.fr}
 
 \author{Khalil Said}
 \address{Université de Lyon, Université Lyon 1, France, Institut Camille Jordan UMR 5208 }
 \email{khalil.said@universite-lyon.fr}

   \title[The Infinite]{Multivariate extensions of expectiles risk measures}

   \begin{abstract}
   This paper is devoted to the introduction and study of a new family of multivariate \textit{elicitable} risk measures. We call the obtained vector-valued measures \textit{multivariate expectiles}. We present the different approaches used to construct our measures. We discuss the coherence properties of these multivariate expectiles. Furthermore, we propose a stochastic approximation tool of these risk measures. 
   \end{abstract}										
   \date{\today}
   \maketitle
   \section*{Introduction}
  The risk measures theory has a rich literature in continuous updating. The mathematical construction of different indicators and methods used in risk management has to be in accordance with the professional practices. In this spirit, the notion of coherence is particularly important in the literature of actuarial science. The coherence of univariate risk measures was introduced in the famous paper of Artzner et al. (1999) \cite{Artz}, a generalization to set-valued risk measures was presented by Jouini et al. (2004) \cite{jouini2004}. The different axioms chosen to characterize the measures coherence are justified by the economic importance of the properties they present in practice.  From that, the coherence is not limited only to these axioms, but concerns all desirable properties from practitioners' point of view.\\
  
  Gneiting (2011) \cite{Gneiting2011} has raised an important issue concerning the statistical coherence of the usual risk measures.  \textit{Elicitability}, according to him, is a natural property that must be satisfied by a measure to ensure the possibility of implementing \textit{backtesting} procedures, which are necessary from a practical point of view. This work has implicitly challenged the notion of coherence, especially for measures that are not elicitable even if they are coherent. Several studies were presented recently on the subject. Some analyze this property in the context of risk measures, others are dedicated to finding a characterization of the measures that are both coherent and elicitable. The \textit{expectiles}  have therefore  appeared in the literature as the only law invariant risk measures that meet this need,  they are elicitable by construction and coherent  for a threshold level range.\\
  
  The concept of elicitability was studied in Ziegel (2014) \cite{Ziegel2014}. It is also examined in a purely mathematical framework by Steinwart et al. (2014) \cite{Cachan2014}. The elicitability of risk measures  is analyzed by Bellini and Bignozzi (2015) \cite{bellini2015elicitable}, and in Wang and Ziegel (2015) \cite{Wang2015}. Expectiles, as risk measures, were the subject of the works presented by Emmer et al. (2013) \cite{Tasche2013}, Bellini et al. (2014) \cite{bellini2014generalized}, and Bellini and Di Bernardino (2015) \cite{belliniElena}.\\
  
  The statistical importance of the elicitability property invites us to reconsider the mathematical construction of some risk measures. Indeed, risk management should be naturally based on a minimization of certain risk quantity or a score, defined according to user needs. From this point of view, it seems useful to integrate this vision since the choice of risk measures, and also to take it in account by insurers in decision making related to internal risk management, like capital allocation. In this paper, we present a construction of a new family of vector-valued risk measures.  It is firstly characterized by its elicitability, and then justified by its economic interpretations. The constructed measures are vectors of the same size as the risks. The literature presents a limited number of such measures, we cite as example, the multidimensional VaR and CVaR introduced in Di Bernardinio (2012)\cite{PhDElena}, and studied in Cousin and Di Bernardinio (2013 ) \cite{cousin2013multivariate}, Cousin and Di Bernardinio (2014) \cite{cousin2014multivariate} and Di Bernardinio et al. (2015) \cite{di2015multivariate}. In practice, such measures can be used in risk allocation and for the measurement of systemic risk.\\
  
  The proposed measures, in this paper, are in general minimizers of strictly convex functions. Therefore, it is possible to use stochastic approximation algorithms to determine the minimum which is the desired measure. We shall use a version of Robbins-Monro's algorithm (1951) )\cite{robbins1951stochastic} based on its multidimensional form proposed by Blum (1954) \cite{blum1954multidimensional}. This method allows approximating our measures in the general case, with quite satisfactory convergence except for the asymptotic levels where the number of observations is limited. We finish this work by a first introductory result on the asymptotic multivariate expectile. The asymptotic behavior deserves a deeper analysis. \\
     
   This paper is organized as follows. In a first section, we present the concept of multivariate elicitability that will constitute the basis of our multivariate extensions of expectiles. Section 2 is a presentation of some constructions of multivariate expectiles. We introduce a generalization of expectiles in higher dimension using norms. We examine in details two expectiles families, \textit{euclidean expectiles} and \textit{matrix expectiles}. We analyze the economic interpretation of the obtained measures, and we discuss what could be the best construction. We then focus on matrix expectiles. We study the different coherence properties that are satisfied by these measures in Section 3. In Section 4, using the Robbins-Monro's stochastic optimization algorithm, we present an approximation method of multivariate expectile in the general case. We give some numerical illustrations of this method. The last section is an introduction to the study of the asymptotic behavior of multivariate expectiles.  
      \section{Elicitable risk measures}
   Since the publication of Gneiting's work \cite{Gneiting2011}, the elicitability notion, coming originally from the decision theory, occupies more importance in the new literature and research studies on the coherence properties of risk measures. In this section, we present a definition of elicitability. We analyze its importance as a property in a context of risk measures and we examine its possible extension in a multivariate case. The aim of this section is to present a general framework that will be the basis of a construction of a vector valued risk measures in the next section.\\
   
   A statistic $T$ is elicitable if it can be written, for any random variable $X$, as a minimizer of a scoring function denoted $s$
    $$T(X)=\underset{x}{\arg\min}~\mathbb{E}_\mathbb{P}[s(x,X)],~~~X\sim \mathbb{P}\/.$$
      The statistical utility of this property can be summarized in two important advantages:
      \begin{itemize}
      \item The ability to compare different statistical methods using the scoring function, and thus give meaning to the \textit{backtesting procedures}, which are very important in finance and insurance.
      \item The ability to produce forecasts and estimations by mean regression in the case of elicitable measures.
      \end{itemize}
      For more details on the statistical importance of the elicitability property, we invite the reader to consult the papers of Gneiting (2011) \cite{ Gneiting2011} and Emmer et al. (2013) \cite{Tasche2013}.
   
    \subsection{Elicitability of risk measures}
   Elicitability is a desirable property of risk measures.  The usual measures are not all elicitable.  Variance and TVaR are not elicitable, proofs are presented in Lambert et al. (2008) \cite{ lambert2008} for the Variance and in Gneiting (2011) \cite{ Gneiting2011} for TVaR. The mean is elicitable because it can be written as 
   \[ \mathbb{E}[X]=\int x\mbox{d}F(x)=\underset{t}{\arg\min}\{\mathbb{E}[(X-t)^2]\} \/.\]
   Quantiles are also elicitable for any threshold level $ \alpha \in [0,1]$. They can be written as a minimizer of the scoring functions $S_\alpha(x,y)=\alpha(y-x)_++(1-\alpha)(x-y)_+$ called Pinball functions 
   \[ q_\alpha(X)=\min\{x: F_X(x)\geq \alpha\}=\underset{x\in \mathbb{R}}{\arg\min}\{\mathbb{E}[\alpha(X-x)_++(1-\alpha)(x-X)_+]\}\/. \]
   The natural question that arises is about the presence of any risk measures which are both coherent and elicitable.  Expectiles are the only coherent as well as elicitable risk measure according to Bellini and Bignozzi (2015) \cite{bellini2015elicitable}.\\
   The expectiles were introduced in the context of statistical regression models by Newey and Powell (1987) \cite{Expect1987}.
   \begin{definition}[Expectiles (Newey et Powell, 1987)]
   	For a random variable $X$ with finite order $2$ moment, the expectile of level $\alpha$ is defined as 
   	\begin{equation}\label{ExpDef}
   	e_\alpha(X)=\arg\min_{x\in \mathbb{R}}\mathbb{E}[\alpha(X-x)_+^{2}+(1-\alpha)(x-X)_+^{2}]\/, 
   	\end{equation}
   	where  $(x)_+=\max(x,0)$.
   \end{definition}
   These measures are elicitable by construction. For $\alpha=1/2$ expectile coincides with the mean. This risk measure is coherent for $\alpha\geqslant 1/2$. For $\alpha<1/2$, the expectile is super-additive and therefore not coherent.\\
   
   The uniqueness of the minimum is guaranteed by the strict convexity of the scoring function. One can also define the expectiles using the optimality condition of the first order, as the unique solution of the equation
   \begin{equation}\label{ExpEq}
   \alpha\mathbb{E}[(X-x)_+]=(1-\alpha)\mathbb{E}[(x-X)_+]\/.
   \end{equation}
   The above equation can also be written as
   $$\frac{1-\alpha}{\alpha}=\frac{\mathbb{E}[(X-x)_+]}{\mathbb{E}[(x-X)_+]}\/.$$
   This makes the economic interpretation of expectile as measure clearer. The expectiles can be seen as threshold that provides a Profits/Loss ratio of value $\frac{1-\alpha}{\alpha}$.\\
   
    The properties of expectiles risk measures have been studied in several recent papers. We recall some of them. The proofs are presented in Emmer et al. (2013) \cite{Tasche2013} and in Bellini and Di Bernardino (2015) \cite{belliniElena}.
     
     \begin{itemize}
     	\item The expectile $e_\alpha(X)$ is a strictly increasing function of $\alpha\in[0,1]$;
     	\item If the distribution of the random variable $X$ is symmetrical with respect to a point $x_0$ then, $$\frac{e_\alpha(X)+e_{1-\alpha}(X)}{2}=x_0\/;$$
     	\item The expectile $e_\alpha(X)$ is a law invariant and positively homogeneous risk measure for all $0<\alpha<1$. It satisfies the invariance by translation;
     	\item The expectiles satisfy  $$e_\alpha(-X)=-e_{1-\alpha}(X)\/,$$
     	for all $\alpha\in[0,1]$;
     	\item The expectile $e_\alpha(X)$ is stochastically strictly monotone function of $X$, which means if $X\leq Y, a.s$ and $\mathbb{P}(X<Y)>0$, then $$e_\alpha(X)<e_\alpha(Y)\/;$$
     	\item The expectiles are sub additive and then coherent for $\alpha\in[1/2,1[$, they are super additive for $\alpha<1/2$;
     	\item The expectiles are additives measures by linear dependence
     	$$\mbox{corr}(X,Y)=1~~\Rightarrow~~e_\alpha(X+Y)=e_\alpha(X)+e_\alpha(Y)\/,$$
     	but they are not comonotonically additive. 
     \end{itemize}  
     The asymptotic behavior of Expectiles is studied by Bellini and Di Bernardino (2015) \cite{belliniElena}.  The second order of this behavior is analyzed by Mao and Yang (2015) \cite{ExpectileFGM}. Bellini et al. (2014) \cite{bellini2014generalized} have introduced the generalized quantiles risk measures, which include Expectiles, defined as a minimizer of an asymmetric error
     \begin{equation}
     x_{\alpha}(X)=\underset{x\in\mathbb{R}}{\arg\min}\{\alpha\mathbb{E}[\Phi_+((X-x)_+)]+(1-\alpha)\mathbb{E}[\Phi_-((X-x)_-)]\}\/,
     \end{equation} 
     where $\Phi_+$ and $\Phi_-$ are a convex scoring functions. Expectiles corresponds to the case $x\rightarrow x^2$ for both functions. Recently, Daouia et al. (2016) \cite{daouia2016estimation} proposed an estimation of the VaR and ES risk measures using Expectiles.\\

   \subsection{Multivariate elicitability}
   
   In order to overcome the lack of elicitability of some usual risk measures, several works have used multivariate versions of elicitability. Lambert (2008) \cite{lambert2008} introduced the concept of indirect elicitability which allows to consider measures such Variance as elicitable via the elicitability of the couple ($\mathbb{E}[X],\mathbb{E}[X^2]$). The same kind of solution has been proposed by Fissler and Ziegel (2015) \cite{ziegel20152} for the elicitability of the couple (VaR, TVaR)  to allow validation of TVaR's Backtesting procedures.\\
   
  The multivariate context of risk management and the multidimensional nature of statistics, make the multivariate generalization of the elicitability property in higher dimension, a natural step of a significant utility in risk modelling. The idea is to build multidimensional measures 
   \begin{equation*}
   T:\mathcal{P}\rightarrow\mathbb{R}^d\/,
   \end{equation*} 
   that can be written in the form
   \[  \arg \inf_{\mathbf{u} \in U} \mathbb{E}[ S(\mathbf{X}, \mathbf{u})]\/,\]
   where $\mathbf{X}$ is a random vector in $\mathbb{R}^d$ and $\mathbf{u}$ a vector of size $d_u$ which can be different from $d$ and $U\subset\mathbb{R}^{d_u}$.\\
   
  Several authors have published important contributions on this subject. The elicitability of vector statistics was studied, using multivariate scoring functions, in Osband (1985) \cite{osband1985}. Lambert et al. (2008) \cite{lambert2008} introduced the notion of \textit {$k$-elicitability}, replacing the elicitability property by a linear combination of elicitable functional.\\
  
Multivariate elicitability also has great importance in the study of learning machines, we cite as examples Frongillo and Kash (2014) \cite{frongillo2014} who study the elicitation of vector-valued measures using the concept of the \textit{separability} of scoring functions. According to this work, a vector-valued statistical is elicitable, if it is the case for each component $T_i:\mathcal{P}\rightarrow\mathbb{R}$,  and in this case,  the vector is elicitable using the sum function of the univariate scoring functions $S_i$.  Separability of scoring functions remains a strong assumption, as example, Osband (1985) shows that there is no separable scoring function that elicits a bivariate quantile. \\

   In this paper, we use multivariate elicitability in a more general context. Our aim is the characterization of elicitability in the case of multidimensional risk measures of the same size as the risk vector. This vision is traduced by Definition \ref{DefMElici}.
   \begin{definition}[Elicitable vectorial measures ]\label{DefMElici}
   	A vector-valued risk measure  $T:\mathcal{P}\rightarrow\mathbb{R}^k$ is elicitable if there exists a scoring function  $s:\mathbb{R}^d\times\mathbb{R}^k\longrightarrow\mathbb{R}$, such that  
   	$$T(\mathbf{X})\in\underset{\mathbf{x}\in U\subset \mathbb{R}^k}{\arg\inf}\mathbb{E}[s(\mathbf{X},\mathbf{x})]\/,$$
   	for all random vector $\mathbf{X}$. 
   \end{definition}
   In a recent paper, Fissler and Ziegel \cite{ziegel20152} studied this definition in general case.  When the scoring function is strictly convex Definition~\ref{DefMElici} becomes simpler  due to  the uniqueness of the minimum. In this paper, we will focus on this case.\\
   
  We must distinguish between the multivariate elicitability as presented in Definition ~\ref{DefMElici} and that of the $k$-elicitability introduced by Lambert et al (2008). For example,  the variance is considered as  $2$-elicitable since  $\mbox{Var}(X)=\mathbb{E}[X^2]-\mathbb{E}[X]^2$  and the measures $\mathbb{E}[X^2]$,$\mathbb{E}[X]$  are elicitable, which does not necessarily imply the multivariate elicitability of the vector $(\mathbb{E}[X^2],\mathbb{E}[X])$.\\ 
  
  In the context of capital allocation, the allocation may be considered as a vector- valued risk measure with $k=d$. In Maume-Deschamps et al. (2016) \cite{maume2016capital}, an axiomatic characterization of multivariate coherence of capital allocation methods is given. It deals with, in particular, the allocation by minimizing multivariate risk indicators, which can be seen as scoring functions. The allocation as measure is therefore elicitable in the sense of Definition ~\ref{DefMElici}.\\\\

   \section{Multivariate extensions of expectiles}
   Following the univariate approach, we present in this section some multivariate constructions of Expectile risk measures.\\
   
    Let $\parallel.\parallel$ be a norm on $\mathbb{R}^d$. We denote by $(\mathbf{X})_+$ the vector $(\mathbf{X})_+=\left((X_1)_+,\ldots,(X_d)_+\right)^T$ and by $(\mathbf{X})_-$ the vector $(\mathbf{X})_-=\left((X_1)_-,\ldots,(X_d)_-\right)^T$. We define the following scoring function
    $$s_\alpha(\mathbf{X},\mathbf{x})=\alpha\parallel(\mathbf{X}-\mathbf{x})_+\parallel^2+(1-\alpha)\parallel(\mathbf{x}-\mathbf{X})_+\parallel^2\/,$$
    for all $\mathbf{x}\in\mathbb{R}^d$.\\
    
    We call \textit{multivariate expectile} any minimizer 
    $$\mathbf{x}^*\in\underset{\mathbf{x}\in\mathbb{R}^d}{\arg\min}~\mathbb{E}[s_\alpha(\mathbf{X},\mathbf{x})]\/.$$ 
     We consider the function $\phi:\mathbb{R}^d\times\mathbb{R}^d\longrightarrow\mathbb{R}$ defined by
    $$\phi(\mathbf{t},\mathbf{x})=\alpha\parallel(\mathbf{t}-\mathbf{x})_+\parallel^2+(1-\alpha)\parallel(\mathbf{x}-\mathbf{t})_+\parallel^2\/.$$
   It is easy to verify that the function $\phi(\mathbf{t},\mathbf{x})$ is strictly convex in $\mathbf{x}$, therefore,  $\mathbb{E}[\phi(\mathbf{X},\mathbf{x})]=\mathbb{E}[s_\alpha(\mathbf{X},\mathbf{x})]$ is also strictly convex in $\mathbf{x}$. In this case, the uniqueness of the minimum is guaranteed, and the multivariate expectile of the vector $\mathbf{X}$ with a confidence level $\alpha$ is defined by
    \begin{equation}\label{ExpMulti}
    	\mathbf{e}_\alpha(\mathbf{X})=\underset{\mathbf{x}\in\mathbb{R}^d}{\arg\min}~\mathbb{E}[\alpha\parallel(\mathbf{X}-\mathbf{x})_+\parallel^2+(1-\alpha)\parallel(\mathbf{x}-\mathbf{X})_+\parallel^2]\/.	
    \end{equation}
    The obtained vector-valued risk measure is elicitable by construction. The choice of a common threshold $\alpha$ for all the components of $\mathbf{X}$ seems natural in insurance contexts, since the accepted risk level must be the same between all off them.\\
 
    In order to illustrate the construction, we present two possible examples of multivariate expectiles families.
    \subsection{Euclidean Expectiles ($L^p$-expectiles)}
    The $L^p$ norms on $\mathbb{R}^d$, $\parallel.\parallel_p,~1\leq p\leq \infty$ are potential candidates to construct multivariate expectiles. Definition \ref{ExpMulti} in this case become 
    \begin{equation*}
    \mathbf{e}_\alpha(\mathbf{X})=\underset{\mathbf{x}\in\mathbb{R}^d}{\arg\min}~\mathbb{E}[\alpha\parallel(\mathbf{X}-\mathbf{x})_+\parallel_p^2+(1-\alpha)\parallel(\mathbf{x}-\mathbf{X})_+\parallel_p^2]\/.	
    \end{equation*} 
    For a norm $\parallel.\parallel_p$ with $p<+\infty$,
    $$ \mathbf{e}_\alpha(\mathbf{X})=\underset{\mathbf{x}\in\mathbb{R}^d}{\arg\min}~\mathbb{E}[\alpha\left(\sum_{i=1}^{d}(X_i-x_i)_+^p\right)^{2/p}+(1-\alpha)\left(\sum_{i=1}^{d}(x_i-X_i)_+^p\right)^{2/p}]\/.$$
    Figure \ref{ScoringLp} is an illustration of the contour lines of the scoring function in the bivariate case $$\psi(z)=\alpha\parallel(z)_+\parallel_p^2+(1-\alpha)\parallel(z)_-\parallel_p^2\/,$$
    for different level of $\alpha$ ($\alpha=0.05,0.5,0.95$) and different euclidean norms ($L_1,L_2$ and $L_{10}$).
    \begin{figure}[H]
    	\center
    	\includegraphics[width=16cm]{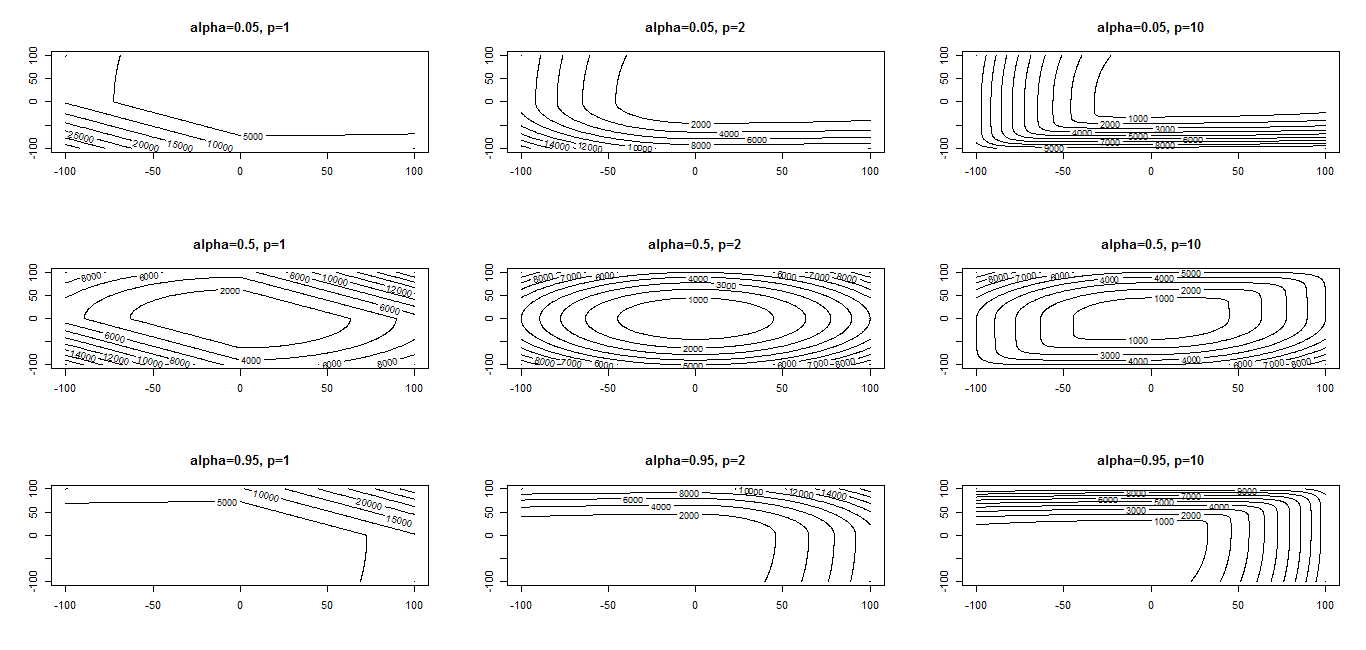} 
    	\caption{$L_p$-expectile: Contour lines of the scoring function}
    	\label{ScoringLp}
    \end{figure}
    When $p=2$, the expectile is not taking into account the dependence. Indeed, the $L^2-$expectile is composed by the marginal univariate expectiles.   
    \begin{align}
    	\mathbf{e}_\alpha(\mathbf{X})&=\underset{\mathbf{x}\in\mathbb{R}^d}{\arg\min}~\mathbb{E}[\alpha\left(\sum_{i=1}^{d}(X_i-x_i)_+^2\right)+(1-\alpha)\left(\sum_{i=1}^{d}(x_i-X_i)_+^2\right)]\nonumber\\\nonumber
    	&=\underset{\mathbf{x}\in\mathbb{R}^d}{\arg\min}~\mathbb{E}[\sum_{i=1}^{d}\left(\alpha(X_i-x_i)_+^2+(1-\alpha)(x_i-X_i)_+^2\right)]\\\nonumber
    	&=(e_\alpha(X_1),\ldots,e_\alpha(X_d))^T   
    	\/.\nonumber
    \end{align}
   Overall, for $p\geq1$, we can determine the optimal solution with the first order necessary conditions for optimality. We have for all $k\in\{1,\ldots,d\}$
    \begin{align*}
    	\nabla_k(\mathbb{E}[s_\alpha(\mathbf{X},\mathbf{x})])=
    	&-2\alpha\mathbb{E}[\frac{\parallel(\mathbf{X}-\mathbf{x})_+\parallel_p}{\parallel(\mathbf{X}-\mathbf{x})_+\parallel^{p-1}_p}(X_k-x_k)^{p-1}_+1\!\!1_{\{X_k>x_k\}}]\\
    	&+2(1-\alpha)\mathbb{E}[\frac{\parallel(\mathbf{x}-\mathbf{X})_+\parallel_p}{\parallel(\mathbf{x}-\mathbf{X})_+\parallel^{p-1}_p}(x_k-X_k)^{p-1}_+1\!\!1_{\{X_k<x_k\}}]\/.
    \end{align*}
    The first order necessary conditions for optimality can be written 
    $$
    \alpha\mathbb{E}[\frac{\parallel(\mathbf{X}-\mathbf{x})_+\parallel_p}{\parallel(\mathbf{X}-\mathbf{x})_+\parallel^{p-1}_p}\frac{\partial(\mathbf{X}-\mathbf{x})_+}{\partial\mathbf{x}}(\mathbf{X}-\mathbf{x})^{p-1}_+]
    =(1-\alpha)\mathbb{E}[\frac{\parallel(\mathbf{x}-\mathbf{X})_+\parallel_p}{\parallel(\mathbf{x}-\mathbf{X})_+\parallel^{p-1}_p}\frac{\partial(\mathbf{x}-\mathbf{X})_+}{\partial\mathbf{x}}(\mathbf{x}-\mathbf{X})^{p-1}_+]\/.
    $$
    The $L^p$-expectile is the unique solution of the equation system 
    \begin{equation}\label{CPOLp}
    	\alpha\mathbb{E}[\frac{\parallel(\mathbf{X}-\mathbf{x})_+\parallel_p}{\parallel(\mathbf{X}-\mathbf{x})_+\parallel^{p-1}_p}(X_k-x_k)^{p-1}_+1\!\!1_{\{X_k>x_k\}}]\\
    	=(1-\alpha)\mathbb{E}[\frac{\parallel(\mathbf{x}-\mathbf{X})_+\parallel_p}{\parallel(\mathbf{x}-\mathbf{X})_+\parallel^{p-1}_p}(x_k-X_k)^{p-1}_+1\!\!1_{\{X_k<x_k\}}]\/,     \end{equation}
    for all $k\in\{1,\ldots,d\}$.\\
    For $2<p<+\infty$, the equations system (\ref{CPOLp}) can be written as follows
    \begin{equation*}
    	\alpha\mathbb{E}[\frac{(X_k-x_k)^{p-1}_+}{\parallel(\mathbf{X}-\mathbf{x})_+\parallel^{p-2}_p}]\\
    	=(1-\alpha)\mathbb{E}[\frac{(X_k-x_k)^{p-1}_-}{\parallel(\mathbf{X}-\mathbf{x})_-\parallel^{p-2}_p}]~~\forall  k\in\{1,\ldots,d\}\/,
    \end{equation*}  
    the dependence is taken into account using the norm. In order to make the economic interpretation more visible, we write the previous system in the following form  
    $$
    \frac{\mathbb{E}[(X_k-x_k)_+\left(\frac{(X_k-x_k)_+}{\parallel(\mathbf{X}-\mathbf{x})_+\parallel_p}\right)^{p-2}]}{\mathbb{E}[(X_k-x_k)_-\left(\frac{(X_k-x_k)_-}{\parallel(\mathbf{X}-\mathbf{x})_-\parallel_p}\right)^{p-2}]}=\frac{1-\alpha}{\alpha}~~\forall  k\in\{1,\ldots,d\}\/.
    $$ 
    The multivariate expectile can be considered as a vector of thresholds that makes a ratio relative gains / relative losses identical between all random variables $X_k$. The multivariate expectile can be considered as a vector of thresholds that makes a ratio relative gains / relative losses identical between all random variables $X_k$.  The adjective “relative” reflects the presence of a weight power of each $X_k$ $\left(\frac{(X_k-x_k)_+}{\parallel(\mathbf{X}-\mathbf{x})_+\parallel_p}\right)^{p-2}$ ($\left(\frac{(X_k-x_k)_-}{\parallel(\mathbf{X}-\mathbf{x})_-\parallel_p}\right)^{p-2}$) in the total aggregated by the norm. This interpretation is more intuitive in the case $p=3$.\\
    
    The case  $p=1$ is quite interesting, a new equivalent definition can be given to the  $L^1$-expectile using the first order condition of optimality (\ref{CPOLp}). The $L^1-$expectile is the unique solution in $\mathbb{R}^d$ of the following system 
    $$
    \alpha\mathbb{E}[\parallel(\mathbf{X}-\mathbf{x})_+\parallel_11\!\!1_{\{X_k>x_k\}}]=(1-\alpha)\mathbb{E}[\parallel(\mathbf{X}-\mathbf{x})_-\parallel_11\!\!1_{\{X_k<x_k\}}], \forall k\in\{1,\ldots,d\}\/.
    $$
    or equivalently
    \begin{equation}\label{CPOL1}
    	\frac{ \mathbb{E}[\parallel(\mathbf{X}-\mathbf{x})_+\parallel_11\!\!1_{\{X_k>x_k\}}]}{\mathbb{E}[\parallel(\mathbf{X}-\mathbf{x})_-\parallel_11\!\!1_{\{X_k<x_k\}}]}=\frac{1-\alpha}{\alpha}, \forall k\in\{1,\ldots,d\}\/,
    \end{equation}
    which can be interpreted as a ratio of an activity participation in the positive scenarios and its participation in the negative ones. The multivariate expectile represents in this case, the vector of marginal levels that make this ration constant for all the activities, and thereby, balance the risk level between the portfolio's components. In insurance, and considering the random variables as risks, the interest will be to the ratio $\beta=(1-\alpha)/\alpha$, which means to increase the level $\alpha$. From this point of view, the multivariate expectile may be a capital allocation tool.\\ 
    
   Equation \ref{CPOL1} can be written using bivariate expected values
    \begin{equation}\label{CPOL1SL}
    	\alpha\sum\limits_{i=1}^{d}\mathbb{E}[(X_i-x_i)_+1\!\!1_{\{X_k>x_k\}}]=(1-\alpha)\sum\limits_{i=1}^{d}\mathbb{E}[(X_i-x_i)_-1\!\!1_{\{X_k<x_k\}}]\/, \forall k\in\{1,\ldots,d\}\/. 
    \end{equation}
    The system (\ref{CPOL1SL}) can also be written using the bivariate Stop-Loss transform functions, introduced and studied for actuarial application by H\"urlimann (2000) \cite{hurlimann2002higher}. Remark that the $L^1$-expectile takes into account only the bivariate dependence.\\
    
    For the maximum norm $\parallel.\parallel_\infty$
    \begin{equation*}
    	\mathbf{e}_\alpha(\mathbf{X})=\underset{\mathbf{x}\in\mathbb{R}^d}{\arg\min}~\mathbb{E}[\alpha\underset{i}{\max}\{(X_i-x_i)^2_+\}+(1-\alpha)\underset{i}{\max}\{(X_i-x_i)^2_-\}]\/.	
    \end{equation*}
    The differentiability of the scoring function is lost. This case is not studied in this paper.
   
    \subsection{Matrix Expectiles ($\Sigma$-expectiles)}
    It is also possible to construct multivariate expectiles using matrices. The idea is to choose a real symmetric matrix 
    $\Sigma$ to define a multivariate scoring function.\\
    Several usual matrices can be used to construct multivariate expectiles, as examples
    \begin{itemize}
    	\item Mahalanobis expectiles: If $\Sigma$ is the inverse matrix of the matrix of covariances, which is equivalent to the use of Mahalanobis distance, we can call the obtained vector Mahalanobis expectile. This constructions is more adapted to the random vectors of elliptic marginal distributions.  
    	\item Correlated expectiles:
    	This construction is based on the use of a correlation matrix. We can also employ any other bivariate dependence measure as coefficients of the construction matrix.
    \end{itemize}
    \begin{definition}[Multivariate matrix expectiles]\label{ExpDefOff}
    	Let $\mathbf{X}=(X_1,\ldots,X_d)^T\in\mathbb{R}^d$ be a random vector such that $\mathbb{E}[|X_iX_j|]<+\infty$ for all $(i,j)\in\{1,\ldots,d\}^2$, and $\Sigma=(\pi_{ij})_{1\leq i,j\leq d}$ a real square matrix  of order $d$, symmetric and positive semi-definite that verifies
    	\begin{enumerate}
    		\item For all $i\in\{1,\ldots,d\}, ~\pi_{ii}=\pi_{i}>0;$
    		\item For all $i,j\in\{1,\ldots,d\}, ~\pi_{ii}\geq\pi_{ij}$.
    	\end{enumerate}
    	We call a $\Sigma-$expectile of $\mathbf{X}$, all vector verifying
    	\begin{equation*}
    	\mathbf{e}^\Sigma_\alpha(\mathbf{X})\in\underset{\mathbf{x}\in\mathbb{R}^d}{\arg\min}~\mathbb{E}[\alpha(\mathbf{X}-\mathbf{x})^T_+\Sigma(\mathbf{X}-\mathbf{x})_++(1-\alpha)(\mathbf{X}-\mathbf{x})^T_-\Sigma(\mathbf{X}-\mathbf{x})_-]\/,	\end{equation*}
    	in the uniqueness case, the multivariate expectile of the vector $\mathbf{X}$ is
    	\begin{equation}\label{Exp-Def}
    	\mathbf{e}^\Sigma_\alpha(\mathbf{X})=\underset{\mathbf{x}\in\mathbb{R}^d}{\arg\min}~\mathbb{E}[\alpha(\mathbf{X}-\mathbf{x})^T_+\Sigma(\mathbf{X}-\mathbf{x})_++(1-\alpha)(\mathbf{X}-\mathbf{x})^T_-\Sigma(\mathbf{X}-\mathbf{x})_-]\/.\end{equation}   
    \end{definition}
    The matrix $\Sigma$ allows to take into account the bivariate dependences between the marginals of $\mathbf{X}$.\\  
   A lower bound of the minimized score is given by  
    $$\mathbb{E}[\alpha(\mathbf{X}-\mathbf{x})^T_+\Sigma(\mathbf{X}-\mathbf{x})_++(1-\alpha)(\mathbf{X}-\mathbf{x})^T_-\Sigma(\mathbf{X}-\mathbf{x})_-]\/\geq \lambda_{\min}\mathbb{E}[\alpha(\mathbf{X}-\mathbf{x})^T_+(\mathbf{X}-\mathbf{x})_++(1-\alpha)(\mathbf{X}-\mathbf{x})^T_-(\mathbf{X}-\mathbf{x})_-]\/,$$
    where $\lambda_{\min}$ is the smallest eigenvalue of $\Sigma$.\\ 
    
    This construction gives an more general framework that encompass several examples, especially the \textit{correlated expectiles}, \textit{Mahalanobis expectile}, and the $L^1$ \textit{euclidean expectile} if  $\Sigma_{ij}=1, \forall i,j$.\\
    
      Under the coefficients positivity assumption of the construction matrix, the multivariate expectile is finally, the unique solution of the following system of equations
    \begin{equation}
    \alpha\sum_{i=1}^{d}\pi_{ki}\mathbb{E}[(X_i-x_i)_+ 1\!\!1_{\{X_k>x_k\}}]
    =(1-\alpha)\sum_{i=1}^{d}\pi_{ki}\mathbb{E}[(x_i-X_i)_+ 1\!\!1_{\{x_k>X_k\}}],~~\forall k\in\{1,\ldots,d\}\/,
    \label{eq1}
    \end{equation}
    equivalent to 
    \begin{equation}
    \alpha=\frac{\sum_{i=1}^{d}\pi_{ki}\mathbb{E}[(x_i-X_i)_+ 1\!\!1_{\{x_k>X_k\}}]}{\sum_{i=1}^{d}\pi_{ki}\left(\mathbb{E}[(X_i-x_i)_+ 1\!\!1_{\{X_k>x_k\}}]+\mathbb{E}[(x_i-X_i)_+ 1\!\!1_{\{x_k>X_k\}}]\right)},~~\forall k\in\{1,\ldots,d\}\/,
    \label{eqME}
    \end{equation} 
    or
    $$\frac{\sum_{i=1}^{d}\pi_{ki}\mathbb{E}[(X_i-x_i)_+ 1\!\!1_{\{X_k>x_k\}}]}{\sum_{i=1}^{d}\pi_{ki}\mathbb{E}[(x_i-X_i)_+ 1\!\!1_{\{x_k>X_k\}}]}=\frac{1-\alpha}{\alpha},~~\forall k\in\{1,\ldots,d\}\/,$$
    it can also be written in the following form 
    \[ \sum_{i=1}^{d}\pi_{ki}x_i\mathbb{E}[Z^\alpha_{X_i,X_k}(x_i,x_k)]
    =
    \sum_{i=1}^{d}\pi_{ki}\mathbb{E}[X_iZ^\alpha_{X_i,X_k}(x_i,x_k)]\/,~~\forall k\in\{1,\ldots,d\}\]
   where   $Z^\alpha_{X_i,X_k}(x_i,x_k)=\alpha1\!\!1_{\{X_i>x_i,X_k>x_k\}}+(1-\alpha)1\!\!1_{\{X_i<x_i,X_k<x_k\}}$.\\
   The optimality system (\ref{eq1}) can also be written 
    \begin{equation}
    \alpha\sum_{i=1}^{d}\pi_{ki}\int_{x_i}^{+\infty}\mathbb{P}\left(X_i>t,X_j>x_j\right)dt=(1-\alpha)\sum_{i=1}^{d}\pi_{ki}\int_{-\infty}^{x_i}\mathbb{P}\left(X_i<t,X_j<x_j\right)d,~~\forall k\in\{1,\ldots,d\}\/.
    \label{eq1-Int}
    \end{equation}
    and
    \begin{equation}
    \alpha\mathbb{E}[ 1\!\!1_{\{\mathbf{X}>\mathbf{x}\}}\Sigma(\mathbf{X}-\mathbf{x})_+]
    =(1-\alpha)\mathbb{E}[(I_d- 1\!\!1_{\{\mathbf{X}>\mathbf{x}\}})\Sigma(\mathbf{X}-\mathbf{x})_-]\/,
    \label{eq1-Matrix}
    \end{equation}
    where $1\!\!1_{\{\mathbf{X}>\mathbf{x}\}}$ is the matrix defined by  $\left(1\!\!1_{\{\mathbf{X}>\mathbf{x}\}}\right)_{ij}=\delta_{ij}1\!\!1_{\{X_i>x_i\}}$, and $I_d$ is the $d$ order identity matrix.\\
    
    Other constructions are possible. We can define a geometric expectiles following the same ideas as for geometric quantiles. The Chaudhuri's approach (1996) \cite{chaudhuri1996} and the Abdous and Theodorescu's one (1992) \cite{abdous1992} are easily adjustable for this purpose.\\ 
    
   The examples presented previously show that the proposed expectiles can be used in various risk management contexts. Finally, choice may differ, according to the practical goal. For example, the use of the correlation matrix reflects a choice of dependence modeling between risks. A matrix composed of tail dependence coefficients emphasizes the consideration of tail dependence. That is why we believe that the question of construction choice should be left open to the users, and it must take into account their applications needs. However, the simplest one remains the matrix expectiles.  In fact, using positive semi-definite matrix with positive coefficients, the strict convexity of the scoring function is guaranteed, and the first order condition of optimality is easy to get as a system of equations linking transformed Stop-Loss functions. For these reasons, the rest of our study is focused on this expectiles family as defined in Definition \ref{ExpDefOff}.
   
   \section{Coherence properties of multivariate expectiles}
  Several coherence properties are satisfied by multivariate expectiles. In this section we present some of these properties. To the best of our knowledge, no axiomatic characterization is proposed for the coherence of a vector-valued risk measures in the literature. 
    \begin{Proposition}\label{Pp1}
    	For any order $2$ random vector $\mathbf{X}=(X_1,\ldots,X_d)^T$ in $\mathbb{R}^d$
    	\begin{enumerate}
    		\item Positive homogeneity : for any non negative real constant $a$, $$\mathbf{e}_\alpha(a\mathbf{X})=a\mathbf{e}_\alpha(\mathbf{X})\/.$$
    		\item Invariance by translation : for any vector $\mathbf{a}=(a_1,\ldots,a_d)^T$ in $\mathbb{R}^d$,  $$\mathbf{e}_\alpha(\mathbf{X}+\mathbf{a})=\mathbf{e}_\alpha(\mathbf{X})+\mathbf{a} \/.$$
    		\item Law invariance : for any order $2$ random vector $Y=(Y_1,\ldots,Y_d)^T$ on $\mathbb{R}^d$ such that $(X_i,X_j)\overset{\mathcal{L}}{=}(Y_i,Y_j)$ for all $(i,j)\in\{1,\ldots,d\}^2$, then  $$\forall \alpha\in[0,1],~~~\mathbf{e}_\alpha(X)=\mathbf{e}_\alpha(Y)\/.$$
    		\item Pseudo-invariance by linear transformations : We denote by $\mathbf{e}_\alpha^\Sigma$ the multivariate expectile obtained using the matrix $\Sigma$.\\
    		For any vectors  $\mathbf{a}=(a_1,\ldots,a_d)^T\in\mathbb{R}^d$ and $\mathbf{b}=(b_1,\ldots,b_d)^T\in(\mathbb{R}^{+*})^d$,  $$\mathbf{e}^\Sigma_\alpha(V\mathbf{X}+\mathbf{a})=V\mathbf{e}^{V\Sigma V}_\alpha(\mathbf{X})+\mathbf{a}\/,$$
    		where $V$ is the diagonal square matrix associated to $\mathbf{b}$, $V=Diag(\mathbf{b}^T)$. 
    	\end{enumerate}
    \end{Proposition}

    \begin{proof} 	
    	The proofs of the first 3 items are straightforward using (\ref{eqME}).\\ 
   For $(4)$, using $(2)$, it is sufficient to prove that $$\mathbf{e}^\Sigma_\alpha(V\mathbf{X})=V\mathbf{e}^{V\Sigma V}_\alpha(\mathbf{X})\/.$$ 
   We have \begin{align*}
   	\mathbf{e}^\Sigma_\alpha(V\mathbf{X})&=\underset{\mathbf{x}\in\mathbb{R}^d}{\arg\min}~\mathbb{E}[\alpha(V\mathbf{X}-\mathbf{x})^T_+\Sigma(V\mathbf{X}-\mathbf{x})_++(1-\alpha)(V\mathbf{X}-\mathbf{x})^T_-\Sigma(V\mathbf{X}-\mathbf{x})_-]\\
   	&=\underset{\mathbf{x}\in\mathbb{R}^d}{\arg\min}~\mathbb{E}[\alpha(\mathbf{X}-V^{-1}\mathbf{x})^T_+V^T\Sigma V(\mathbf{X}-V^{-1}\mathbf{x})_++(1-\alpha)(\mathbf{X}-V^{-1}\mathbf{x})^T_-V^T\Sigma V(\mathbf{X}-V^{-1}\mathbf{x})_-]\\
   	&=V\mathbf{e}^{V^T\Sigma V}_\alpha(\mathbf{X})=V\mathbf{e}^{V\Sigma V}_\alpha(\mathbf{X})\/,
   	\end{align*}
   	because $V$ is symmetric.
   \end{proof}
  \begin{Proposition}[Symmetry with respect to $\alpha$]\label{P-SPalpha}
  	For any order $2$ random vector $\mathbf{X}=(X_1,\ldots,X_d)^T$ in $\mathbb{R}^d$, and any $\alpha\in[0,1]$,  $$\mathbf{e}_\alpha(-\mathbf{X})=-\mathbf{e}_{1-\alpha}(\mathbf{X}) \/.$$
  \end{Proposition} 
  \begin{proof}
  	We denote by $\mathbf{x}^*=\mathbf{e}_{1-\alpha}(\mathbf{X})$ the multivariate expectile with threshold $(1-\alpha)$ associated to the random vector $\mathbf{X}$ and  $\mathbf{x}=\mathbf{e}_\alpha(-\mathbf{X})$ the multivariate expectile threshold $\alpha$ associated to $-\mathbf{X}$. Using the optimality condition (\ref{eq1}), $\mathbf{x}$ is the unique solution of the following equation system $$\alpha\sum_{i=1}^{d}\pi_{ki}\mathbb{E}[(-X_i-x_i)_+ 1\!\!1_{\{-X_k>x_k\}}]
  	=(1-\alpha)\sum_{i=1}^{d}\pi_{ki}\mathbb{E}[(x_i+X_i)_+ 1\!\!1_{\{x_k>-X_k\}}],~~\forall k\in\{1,\ldots,d\}\/,$$
  	rewritten as 
  	$$\alpha\sum_{i=1}^{d}\pi_{ki}\mathbb{E}[(-x_i-X_i)_+ 1\!\!1_{\{X_k<-x_k\}}]
  	=(1-\alpha)\sum_{i=1}^{d}\pi_{ki}\mathbb{E}[(X_i-(-x_i))_+ 1\!\!1_{\{-x_k<X_k\}}],~~\forall k\in\{1,\ldots,d\}\/.$$
  	From that we deduce $-\mathbf{x}$ is the unique solution of the optimality condition corresponding $\mathbf{e}_{1-\alpha}(\mathbf{X})$, so $-\mathbf{x}=\mathbf{x}^*$.
\end{proof}
  \begin{Proposition}[Impact of independence]\label{P-PCI} For a random vector $\mathbf{X}=(X_1,\ldots,X_d)^T$, such that  $\exists i\in\{1,\ldots,d\}$ with $\pi_{ij}=0, \forall j\in\{1,\ldots,d\}\backslash\{i\}$  
  	$$\mathbf{e}^i_\alpha(\mathbf{X})=e_\alpha(X_i)\/,$$
  	where $\mathbf{e}^i_\alpha(\mathbf{X})$ is the  $i^{\mbox{th}}$ coordinate of $\mathbf{e}_\alpha(\mathbf{X})$, and $e_\alpha(X_i)$ is the univariate expectile associated to the random variable $X_i$.
  \end{Proposition}
   \begin{proof} Since $\pi_{ij}=0$ $\forall j\neq i$, the $i^{\mbox{th}}$ equation of the optimality system (\ref{eqME}) is
  	$$\alpha\mathbb{E}[(X_i-\mathbf{e}^i_\alpha(\mathbf{X}))_+]=(1-\alpha)\mathbb{E}[(X_i-\mathbf{e}^i_\alpha(\mathbf{X}))_-]\/,$$ 
  	which is the optimality condition that define the univariate expectile of $X_i$.
  \end{proof}
  \begin{Proposition}[The support stability]
  	For any order $2$ random vector $\mathbf{X}=(X_1,\ldots,X_d)^T$ in $\mathbb{R}^d$ $$\mathbf{e}^i_\alpha(\mathbf{X})\in[x^i_I, x_F^i] \/.$$
  	where $\mathbf{e}^i_\alpha(\mathbf{X})$ is the $i^{\mbox{th}}$ coordinate of $\mathbf{e}_\alpha(\mathbf{X})$,  $x^i_I=essInf(X_i)$ and $x^i_F=essSup(X_i)$.
  \end{Proposition}
   \begin{proof} In the case of infinite marginal support, the result is trivial. Let us assume the existence of $k\in\{1,\ldots,d\}$ such that $x^k_F<+\infty$ and $\mathbf{e}^k_\alpha(\mathbf{X})>x^k_F$. The multivariate expectile is defined as 
  	\begin{equation*}
  		\mathbf{e}_\alpha(\mathbf{X})=\underset{\mathbf{x}\in\mathbb{R}^d}{\arg\min}~\mathbb{E}[\alpha(\mathbf{X}-\mathbf{x})^T_+\Sigma(\mathbf{X}-\mathbf{x})_++(1-\alpha)(\mathbf{X}-\mathbf{x})^T_-\Sigma(\mathbf{X}-\mathbf{x})_-]\/,	\end{equation*}
  	and the minimum in this case is
  	\begin{align*}
  		\underset{\mathbf{x}\in\mathbb{R}^d}{\min}~\mathbb{E}[s_\alpha(\mathbf{X},\mathbf{x})]&=
  		\underset{\mathbf{x}\in\mathbb{R}^d}{\min}~\mathbb{E}[\alpha(\mathbf{X}-\mathbf{x})^T_+\Sigma(\mathbf{X}-\mathbf{x})_++(1-\alpha)(\mathbf{X}-\mathbf{x})^T_-\Sigma(\mathbf{X}-\mathbf{x})_-]\\
  		&=(1-\alpha)\left(\pi_k\mathbb{E}[(\mathbf{e}^k_\alpha(\mathbf{X})-X_k)^2_+]+\sum_{1\leq i\leq d,i\neq k}^{}\pi_{ik}\mathbb{E}[(\mathbf{e}^i_\alpha(\mathbf{X})-X_i)_+(\mathbf{e}^k_\alpha(\mathbf{X})-X_k)_+]\right)\\&~+(1-\alpha)\sum_{\underset{1\leq j\leq d,j\neq k}{1\leq i\leq d,i\neq k}}^{}\pi_{ij}\mathbb{E}[(\mathbf{e}^i_\alpha(\mathbf{X})-X_i)_+(\mathbf{e}^j_\alpha(\mathbf{X})-X_j)_+]\\&~+\alpha\sum_{\underset{1\leq j\leq d,j\neq k}{1\leq i\leq d,i\neq k}}^{}\pi_{ij}\mathbb{E}[(X_i-\mathbf{e}^i_\alpha(\mathbf{X}))_+(X_j-\mathbf{e}^j_\alpha(\mathbf{X}))_+]\/,
  	\end{align*}
  	because $\mathbb{E}[(X_k-\mathbf{e}^k_\alpha(\mathbf{X}))_+(X_i-\mathbf{e}^i_\alpha(\mathbf{X}))_+]=0$ for all $i\in\{1,\ldots,d\}$.\\
  	The function $x\longrightarrow(1-\alpha)\left(\pi_k\mathbb{E}[(x-X_k)^2_+]+\sum_{1\leq i\leq d,i\neq k}^{}\pi_{ik}\mathbb{E}[(\mathbf{e}^i_\alpha(\mathbf{X})-X_i)_+(x-X_k)_+]\right)$ is non decreasing in $x$, considering  $\mathbf{e}^k_\alpha(\mathbf{X})>x^*>x^i_F$, for $i\in\{1,\ldots,d\}\backslash\{k\}$  $$\mathbb{E}[(X_k-x^*)^2_+]=\mathbb{E}[(X_k-x^*)_+(X_i-\mathbf{e}^i_\alpha(\mathbf{X}))_+]=0\/,$$
  	we have for  $\mathbf{x}^*=(\mathbf{e}^1_\alpha(\mathbf{X}),\ldots,\mathbf{e}^{k-1}_\alpha(\mathbf{X}),x^*,\mathbf{e}^{k+1}_\alpha(\mathbf{X}),\ldots,\mathbf{e}^d_\alpha(\mathbf{X}))$ 
  	\begin{align*}
  		\mathbb{E}[s_\alpha(\mathbf{X},\mathbf{x}^*)]&=(1-\alpha)\left(\pi_k\mathbb{E}[(x^*-X_k)^2_+]+\sum_{1\leq i\leq d,i\neq k}^{}\pi_{ik}\mathbb{E}[(\mathbf{e}^i_\alpha(\mathbf{X})-X_i)_+(x^*-X_k)_+]\right)\\&~+(1-\alpha)\sum_{\underset{1\leq j\leq d,j\neq k}{1\leq i\leq d,i\neq k}}^{}\pi_{ij}\mathbb{E}[(\mathbf{e}^i_\alpha(\mathbf{X})-X_i)_+(\mathbf{e}^j_\alpha(\mathbf{X})-X_j)_+]\\&~+\alpha\sum_{\underset{1\leq j\leq d,j\neq k}{1\leq i\leq d,i\neq k}}^{}\pi_{ij}\mathbb{E}[(X_i-\mathbf{e}^i_\alpha(\mathbf{X}))_+(X_j-\mathbf{e}^j_\alpha(\mathbf{X}))_+]\\
  		&<\mathbb{E}[s_\alpha(\mathbf{X},\mathbf{e}_\alpha(\mathbf{X}))]=\underset{\mathbf{x}\in\mathbb{R}^d}{\min}~\mathbb{E}[s_\alpha(\mathbf{X},\mathbf{x})]\/,
  	\end{align*}
  	which is absurd. We deduce that
  	 $$\mathbf{e}^k_\alpha(\mathbf{X})\leq x_F^k, ~\forall k\in\{1,\ldots,d\}\/,$$
  	 and using the symmetry by $\alpha$ property, we get 
  	$$\mathbf{e}^k_\alpha(\mathbf{X})\geq x_I^k, ~\forall k\in\{1,\ldots,d\}\/.$$ 
  \end{proof}
  \begin{Proposition}[Strong intern monotony]\label{Prop-MI}
  	We consider the case of multivariate expectile constructed using a symmetric positive-definite matrix $\Sigma=(\pi_{ij})_{1\leq i\leq j \leq d}$ of a positive coefficients.\\
  	For any order $2$ random vector $\mathbf{X}=(X_1,\ldots,X_d)^T$ in $\mathbb{R}^d$, if there exists a couple $(i,j)\in\{1,\ldots,d\}^2$ such that $X_i\leq X_j, a.s$, $\pi_{i}\leq\pi_{j}$,  $\pi_{ik}\leq\pi_{jk}$ for all $k\in\{1,\ldots,d\}\backslash\{i,j\}$, and the couples $(X_i,X_k)$ and $(X_j,X_k)$ have the same copula for all $k\in\{1,\ldots,d\}\backslash\{i,j\}$,
  	then 
  	$$\mathbf{e}_\alpha^i(\mathbf{X})\leq\mathbf{e}_\alpha^j(\mathbf{X})\/,$$
  	for all $\alpha \in [0,1]$.
  \end{Proposition}
   \begin{proof}
  	Let $(X_1,\ldots,X_d)^T$ be a random vector in $\mathbb{R}^d$. We suppose the existence of $i$ and $j$ such that $X_i\leq X_j, a.s$. We denote by $l^\alpha_{X_i,X_j}$ for all $(i,j)\in\{1,\ldots,d\}^2$ the functions
  	$$l^\alpha_{X_i,X_j}(x_i,x_j)=\alpha\mathbb{E}[(X_i-x_i)_+ 1\!\!1_{\{X_j>x_j\}}]-(1-\alpha)\mathbb{E}[(X_i-x_i)_- 1\!\!1_{\{X_j<x_j\}}]\/,$$
  	for all $(x_i,x_j)\in\mathbb{R}^2$, and by  $l^\alpha_{X_i}$ the function $l^\alpha_{X_i}(x_i)=l^\alpha_{X_i,X_i}(x_i,x_i)$.\\ The multivariate expectile is the unique solution of System (\ref{eq1}) that can be written using the functions  $l^\alpha$ 
  	\begin{equation}\label{CPO-fonctionsL}
  		\sum_{i=1}^{d}\pi_{ki}l^\alpha_{X_i,X_k}(x_i,x_k)=0~~\forall k\in\{1,\ldots,d\}\/.
  	\end{equation} 
  	The functions $l^\alpha_{X_i}$ are non increasing for all $i\in\{1,\ldots,d\}$, and $l^\alpha_{X_i,X_j}(x_i,x_j)$ are non increasing in $x_i$ for all $(i,j)\in\{1,\ldots,d\}^2$. 
  	On another side, for all $(i,j)\in\{1,\ldots,d\}^2$
  	$$l^\alpha_{X_i,X_j}(x_i,x_j)=\mathbb{E}[\left(\alpha(X_i-x_i)_++(1-\alpha)(X_i-x_i)_-\right)1\!\!1_{\{X_j\geq x_j\}}]-(1-\alpha)\mathbb{E}[(X_i-x_i)_-]\/,$$
  	functions $l^\alpha_{X_i,X_j}$ are then non increasing in  $x_j$ too.\\ 
  	We deduce that if $X_i\leq_{}X_j, a.s$, then for all $X_k$ that has the same bivariate copula with both $X_i$ and $X_j$, we have
  	$$l^\alpha_{X_k,X_i}(x_k,x)\leq l^\alpha_{X_k,X_j}(x_k,x),~~\forall (x_k,x)\/.$$
  	Now, we suppose that
  	$x_j=\mathbf{e}_\alpha^j(\mathbf{X})<\mathbf{e}_\alpha^i(\mathbf{X})=x_i\/.$\\
  We have
  	\begin{equation*}
  		l^\alpha_{X_k,X_i}(x_k,x_i)\leq l^\alpha_{X_k,X_i}(x_k,x_j) \leq l^\alpha_{X_k,X_j}(x_k,x_j)\/, \forall k\in\{1,\ldots,d\}\setminus\{i,j\}\/. 
  	\end{equation*}
  	The coefficients of $\Sigma$ verify 
  	$$0 \leq\pi_{ik}\leq\pi_{jk}, \forall k\in\{1,\ldots,d\}\/,$$ 
  	so, from the system of optimality (\ref{CPO-fonctionsL}), we can deduce that 
  	\begin{equation}\label{OSeq2}
  		\sum_{k=1,k\neq i, k\neq j}^{d}\pi_{ik}l^\alpha_{X_k,X_i}(x_k,x_i)\leq \sum_{k=1,k\neq i, k\neq j}^{d}\pi_{jk} l^\alpha_{X_k,X_j}(x_k,x_j)\/.
  	\end{equation} 
  	Since $x_i>x_j$, the almost sure dominance $X_i\leq X_j, a.s$ implies   
  	$$1\!\!1_{\{X_i\geq x_i,X_j\leq x_j\}}=0\/.$$
  	We deduce from that
  	$$\mathbb{E}[(X_i-x_i)_+]=\mathbb{E}[(X_i-x_i)_+ 1\!\!1_{\{X_j>x_j\}}]\mbox{ and }  \mathbb{E}[(X_j-x_j)_-]=\mathbb{E}[(X_j-x_j)_- 1\!\!1_{\{X_i<x_i\}}]\/,$$
  	hence, we obtain
  	$$ l^\alpha_{X_i}(x_i)< l^\alpha_{X_i,X_j}(x_i,x_j)\mbox{ and } l^\alpha_{X_j,X_i}(x_j,x_i)<l^\alpha_{X_j}(x_j)\/,$$
  	which gives us
  	\begin{equation}\label{OS-eq4}
  		\pi_{ij}l^\alpha_{X_i}(x_i)+\pi_{ij}l^\alpha_{X_j,X_i}(x_j,x_i) < \pi_{jj}l^\alpha_{X_j}(x_j)+ \pi_{ij}l^\alpha_{X_i,X_j}(x_i,x_j)\/. 
  	\end{equation}
  	We have
  	$$l^\alpha_{X_i}(x)\leq l^\alpha_{X_j}(x),~~\forall x,$$
  	because the function  $t\longrightarrow\alpha(t-x)_+-(1-\alpha)(x-t)_+$ is non decreasing in $t$. Therefore
  	\begin{equation}\label{OSeq3}
  		(\pi_{i}-\pi_{ij})l^\alpha_{X_i}(x_i)\leq (\pi_{j}-\pi_{ij})l^\alpha_{X_j}(x_j)\/,
  	\end{equation}
  	because  $0\leq\pi_{i}-\pi_{ij}\leq\pi_{j}-\pi_{ij}$.\\
  	Finally, form (\ref{OSeq3}) and (\ref{OS-eq4}) 
  	\begin{equation}
  		\pi_{i}l^\alpha_{X_i}(x_i)+\pi_{ij}l^\alpha_{X_j,X_i}(x_j,x_i) < \pi_{j}l^\alpha_{X_j}(x_j)+ \pi_{ij}l^\alpha_{X_i,X_j}(x_i,x_j)\/, 
  	\end{equation}
  	which is contradictory with (\ref{OSeq2}) and the system of optimality (\ref{CPO-fonctionsL}).  
  \end{proof}

  \begin{Proposition}[Derivatives by respect to $\alpha$]\label{PropSensVariations}
  	Let $\mathbf{X}=(X_1,\ldots,X_d)^T$ be a random vector of a continuous distribution in $\mathbb{R}^d$. We consider a multivariate expectile constructed using a positive semi-definite matrix, then the vector $X_{\partial_\alpha}=(\frac{\partial \mathbf{e}_\alpha^1(\mathbf{X})}{\partial \alpha},\ldots,\frac{\partial \mathbf{e}_\alpha^d(\mathbf{X})}{\partial \alpha})^T$ composed of the derivatives by $\alpha$ satisfies the following system of equations 
  	\begin{equation*}
  		B_k\frac{\partial x_k}{\partial\alpha}+\sum_{i=1}^{d}\gamma_{ki}\frac{\partial x_i}{\partial\alpha}=A_k, ~~~\forall k\in\{1,\ldots,d\}\/,
  	\end{equation*}
  	where $$A_k=\sum_{i=1}^{d}\pi_{ki}\mathbb{E}[(X_i-x_i)_+1\!\!1_{\{X_k>x_k\}}]+\sum_{i=1}^{d}\pi_{ki}\mathbb{E}[(x_i-X_i)_+1\!\!1_{\{X_k<x_k\}}]\geq0\/,$$
  	and 
  	$$B_k=f_{X_k}(x_k)\sum_{i=1,i\ne k}^{d}\pi_{ki}\left(\alpha\mathbb{E}[(X_i-x_i)_+\mid X_k=x_k]+(1-\alpha)\mathbb{E}[(x_i-X_i)_+\mid X_k=x_k]\right)\/,$$
  	for all $k\in\{1,\ldots,d\}$, and
  	$$\gamma_{ki}=\alpha\pi_{ki}\mathbb{P}(X_i>x_i,X_k>x_k)+(1-\alpha)\pi_{ki}\mathbb{P}(X_i<x_i,X_k<x_k),~\forall (i,k)\in\{1,\ldots,d\}^2\/.$$ 	
  \end{Proposition}
   \begin{proof} We denote  $x_i=\mathbf{e}^i_\alpha(\mathbf{X})$. From the system of optimality (\ref{eqME}) 
  	\[ \alpha\sum_{i=1}^{d}\pi_{ki}\mathbb{E}[(X_i-x_i)_+1\!\!1_{\{X_k>x_k\}}]
  	=
  	(1-\alpha)\sum_{i=1}^{d}\pi_{ki}\mathbb{E}[(x_i-X_i)_+1\!\!1_{\{X_k<x_k\}}]\/,~~\forall k\in\{1,\ldots,d\}\/.\]
  	 Thus $\forall ~k=1,\ldots,d$
  	\begin{align}\label{eq-Croissance}
  		A_k&=\frac{\partial x_k}{\partial\alpha}\sum_{i=1}^{d}\pi_{ki}\left((1-\alpha)\frac{\partial }{\partial x_k}\mathbb{E}[(x_i-X_i)_+1\!\!1_{\{X_k<x_k\}}]-\alpha\frac{\partial }{\partial x_k}\mathbb{E}[(X_i-x_i)_+1\!\!1_{\{X_k>x_k\}}]\right)\nonumber \\
  		&+\sum_{i=1,i\neq k}^{d}\pi_{ki}\frac{\partial x_i}{\partial\alpha}\left((1-\alpha)\frac{\partial }{\partial x_i}\mathbb{E}[(x_i-X_i)_+1\!\!1_{\{X_k<x_k\}}]-\alpha\frac{\partial }{\partial x_i}\mathbb{E}[(X_i-x_i)_+1\!\!1_{\{X_k>x_k\}}]\right)
  		\/,
  	\end{align}
  	where $A_k=\sum_{i=1}^{d}\pi_{ki}\mathbb{E}[(X_i-x_i)_+1\!\!1_{\{X_k>x_k\}}]+\sum_{i=1}^{d}\pi_{ki}\mathbb{E}[(x_i-X_i)_+1\!\!1_{\{X_k<x_k\}}]$.\\
  	Moreover, $\forall i\ne k =1,\ldots,d$
  	\begin{align*}
  		\frac{\partial }{\partial x_k}\mathbb{E}[(X_i-x_i)_+1\!\!1_{\{X_k>x_k\}}]
  		&=\int_{x_i}^{+\infty}-f_{X_k}(x_k)\mathbb{P}(X_i>t\mid X_k=x_k)dt\\&=-f_{X_k}(x_k)\mathbb{E}[(X_i-x_i)_+\mid X_k=x_k]\/,
  	\end{align*}
  	and   
  	$$
  	\frac{\partial }{\partial x_k}\mathbb{E}[(x_i-X_i)_+1\!\!1_{\{X_k<x_k\}}]=f_{X_k}(x_k)\mathbb{E}[(x_i-X_i)_+\mid X_k=x_k]\/,
  	$$  
  	
  	$$	\frac{\partial }{\partial x_i}\mathbb{E}[(X_i-x_i)_+1\!\!1_{\{X_k>x_k\}}]=\mathbb{P}(X_i>x_i,X_k>x_k),~\forall (i,k)\in\{1,\ldots,d\}^2\/,$$
  	
  	$$\frac{\partial }{\partial x_i}\mathbb{E}[(x_i-X_i)_+1\!\!1_{\{X_k<x_k\}}]=\mathbb{P}(X_i<x_i,X_k<x_k),~\forall (i,k)\in\{1,\ldots,d\}^2\/.$$
  	From (\ref{eq-Croissance}), we deduce that $X_{\partial_\alpha}=(\frac{\partial x_1}{\partial \alpha},\ldots,\frac{\partial x_d}{\partial \alpha})^T$ satisfies the announced equation system. 
  	
  \end{proof}
   In this section, we have shown that multivariate expectiles satisfy a set of desirable properties of multivariate risk measures that confirms their potential utility. 
   \section{Stochastic estimation}
   In general, the multivariate expectiles cannot be calculated directly, but their estimation is possible using noisy observations. In this section, we present a stochastic approximation method which is adapted to the multivariate expectiles. We mainly use some usual stochastic optimization and root finding algorithms.\\
   
  The multivariate expectile obtained using a positive semi-definite matrix $\Sigma=(\pi_{ij})_{1\leq i\leq j\leq d}$, is the unique solution of system of equations (\ref{eq1})   
 \begin{equation}\nonumber
     \alpha\sum_{i=1}^{d}\pi_{ki}\mathbb{E}[(X_i-x_i)_+ 1\!\!1_{\{X_k>x_k\}}]
     =(1-\alpha)\sum_{i=1}^{d}\pi_{ki}\mathbb{E}[(x_i-X_i)_+ 1\!\!1_{\{x_k>X_k\}}],~~\forall k\in\{1,\ldots,d\}\/,
     \end{equation}
     which has also the form 
     $$
     \phi(\mathbf{x})=\mathbb{E}[\Phi(\mathbf{x},\mathbf{X})]=0\/,
     $$
     where $\Phi(.,\mathbf{X})$ is a function of $\mathbb{R}^d$ on $\mathbb{R}^d$, with $$\Phi_k(\mathbf{x},\mathbf{X})=\sum_{i=1}^{d}\pi_{ki}\left(\alpha(X_i-x_i)_+ 1\!\!1_{\{X_k>x_k\}}-(1-\alpha)(x_i-X_i)_+ 1\!\!1_{\{x_k>X_k\}}\right),~k\in\{1,\ldots,d\}\/.$$
     The stochastic approximation methods are iterative algorithms of the following form 
     $$\mathbf{x}_{n+1}=\mathbf{x}_{n}+\gamma_n\Phi(\mathbf{x}_n,\mathbf{X}_{n+1})\/,$$ 
    where $(\gamma_n)$ is a deterministic sequence of steps which satisfies some further specified conditions, and $(\mathbf{X}_{n})$ is a sequence of independent and identically distributed random vectors, obtained from the same distribution of a generic random variable  $\mathbf{X}$.

      \subsection{Robbins-Monro's algorithm}
      To obtain the multivariate expectile, we use the Robbins-Monro's algorithm. We denote by $\mathbf{x}$
      the desired expectile. The idea of the algorithm is to define a sequence ($\mathbf{x}_n$)
      $$\mathbf{x}_{n+1}=\mathbf{x}_n+\gamma_n\mathbf{Z}_{n+1}\/,$$
      where $\mathbf{x}_0\in\mathbb{R}^d$ is the starting point, and  $\mathbf{Z}_{n+1}$ is an observed random variable which satisfies $\phi(\mathbf{x}_n)=\mathbb{E}[\mathbf{Z}_{n+1}\mid \mathcal{F}_n]$ with $\mathcal{F}_n$ the $\sigma-$algebra of information present in the time $n$
      $$\mathcal{F}_n=\sigma(\mathbf{x}_0,\mathbf{Z}_{1},\mathbf{x}_1,\ldots,\mathbf{x}_{n-1},\mathbf{Z}_{n})\/.$$
      
      Here, we use the version of the Robbins-Monro's theorem as presented in Fraysse (2013) \cite{fraysse2013estimation} (Theorem 1.3.1).
      
      \begin{theorem}[Robbins-Monro]\label{RM-Theo}
      	Under the following assumptions:
      	\begin{enumerate}
      		\item $\phi$ is a continuous function; 
      		\item For all $\mathbf{x}\neq \mathbf{x}^*$, $$(\mathbf{x}-\mathbf{x}^*)^T\phi(\mathbf{x})<0\/;$$
      		\item For al $n\geq0$, 
      		$$\mathbb{E}[\mathbf{Z}_{n+1}\mid \mathcal{F}_n]=\phi(\mathbf{x}_n),~\mbox{a.s.}\/;$$
      		\item There exists $K>0$ such that
      		$$\mathbb{E}[\parallel\mathbf{Z}_{n+1}\parallel^2\mid \mathcal{F}_n]\leq K\left(1+\parallel\mathbf{x}_n-\mathbf{x}^*\parallel^2\right)~\mbox{a.s.}\/;$$
      		\item The sequence $(\gamma_n)$ is decreasing to $0$ and satisfies $$\sum_{n=0}^{+\infty}\gamma_n=+\infty~~\mbox{et}~~\sum_{n=0}^{+\infty}\gamma_n^2<+\infty\/,$$
      	\end{enumerate}
      		the sequence $(\mathbf{x}_n)$ defined by 
      		$$\mathbf{x}_{n+1}=\mathbf{x}_n+\gamma_n\mathbf{Z}_{n+1}\/,$$
      		converges almost surely to $\mathbf{x}^*$, solution of $\phi(\mathbf{x}^*)=\mathbf{0}$.
      	
      \end{theorem}
      The first two assumptions deal with the regularity of $\phi$, they guarantee the solution's uniqueness. Assumptions 3 and 4 are related to the observations sequence, they respectively guarantee its close distance to the exact value and its variance control. The last assumption on the sequence of steps is required to achieve convergence.  Other modified versions of the algorithm are proposed in the literature when these assumptions are not satisfied.
      
      \subsection{Application for multivariate expectiles} 
      For a random vector $\mathbf{X}$ of continuous marginal distributions, the function $\phi$  defined by $\phi(\mathbf{x})=\mathbb{E}[\Phi(\mathbf{x},\mathbf{X})]$ is clearly continuous.\\
      On the other hand, the multivariate expectile $\mathbf{x}^*$ is the minimum of a strict convex function on $\mathbb{R}^d$ of gradient $-2\phi(\mathbf{x})$, it satisfies then for all $\mathbf{x}\in\mathbb{R}^d\setminus\{\mathbf{x}^*\}$
      $$<-2\phi(\mathbf{x})+2\phi(\mathbf{x}^*),\mathbf{x}-\mathbf{x}^*>~=~<-2\phi(\mathbf{x}),\mathbf{x}-\mathbf{x}^*>~>~0\/,$$ we deduce from that, for all $\mathbf{x}\neq \mathbf{x}^*$, $$(\mathbf{x}-\mathbf{x}^*)^T\phi(\mathbf{x})<0\/.$$
      The assumptions of Robbins-Monro's Theorem~\ref{RM-Theo} on $\phi$ thus satisfied.
      
     Consider the sequence of independent and identically distributed random vectors $(\mathbf{X}_n)\overset{\mathcal{L}}{=}\mathbf{X}$. We define the sequence of observations $(Z_n)$ for each iteration $n\geq1$ by 
      $$\mathbf{Z}_n^{(k)}=\sum_{i=1}^{d}\pi_{ki}\left(\alpha(\mathbf{X}_{n}^{(i)}-\mathbf{x}_{n-1}^{(i)})_+ 1\!\!1_{\{\mathbf{X}_{n}^{(k)}>\mathbf{x}_{n-1}^{(k)}\}}-(1-\alpha)(\mathbf{x}_{n-1}^{(i)}-\mathbf{X}_{n}^{(i)})_+ 1\!\!1_{\{\mathbf{X}_{n}^{(k)}<\mathbf{x}_{n-1}^{(k)}\}}\right)$$ 
      for all $k\in\{1,\ldots,d\}$, where $\mathbf{x}^{(i)}$ denote the $i^{\mbox{th}}$ coordinate of $\mathbf{x}$. This previous sequence satisfies for all $k\in\{1,\ldots,d\}$ and all $n\geq0$ 
      $$\mathbb{E}[\mathbf{Z}_{n+1}^{(k)}\mid \mathcal{F}_n]=\phi_k(\mathbf{x}_n),~\mbox{a.s.}\/,$$
      hence  $$\forall n\geq0,~~\mathbb{E}[\mathbf{Z}_{n+1}\mid \mathcal{F}_n]=\phi(\mathbf{x}_n),~\mbox{a.s.}\/.$$
      Using the triangle inequality of the 
      absolute value as norm, and then the Hölder's inequality ($p=q=2$), by setting $K=2\max\left(\alpha,1-\alpha\right)^2\left(\sum_{k=1}^{d}\sum_{i=1}^{d}|\pi_{ik}|^2\right)\max\left(\parallel\mathbf{X}-\mathbf{x}^{*}\parallel^2,1\right)$ we obtain
      $$\mathbb{E}[\parallel\mathbf{Z}_{n+1}\parallel^2\mid \mathcal{F}_n]\leq K\left(1+\parallel\mathbf{x}_n-\mathbf{x}^*\parallel^2\right)~\mbox{a.s.}\/.$$
      Robbins-Monro’s Theorem~\ref{RM-Theo} assumptions on the observations sequence are satisfied. The theorem is then relevant for any sequence of steps $(\gamma_n)$, decreasing to $0$, chosen such that 
       $$\sum_{n=0}^{+\infty}\gamma_n=+\infty~~\mbox{et}~~\sum_{n=0}^{+\infty}\gamma_n^2<+\infty\/.$$
       A natural choice of the sequence of steps is $\gamma_n=1/n$ or $1/n^\kappa$ with $1/2<\kappa<1$. More generally, the choice of this sequence can be made  by adjusting the constants $\kappa,a,b$ such that $\gamma_n=a/(b+n)^\kappa$. Using some numerical illustrations, we will discuss the impact of this choice on the speed of algorithm's convergence.   
   \subsection{Numerical illustrations} 
    We consider a simple bivariate exponential model. $X_1\sim\exp(\beta_1)$ and $X_2\sim\exp(\beta_2)$. Firstly, we 
    examine under independence assumption the $L_1$-expectile case ($\pi_{ij}=1,~\forall(i,j)\in\{1,2\}^2$). The optimality system is explicit and composed of the two following equations 
    $$(2\alpha-1)\frac{1}{\beta_i}e^{-\beta_ix_i}-(1-\alpha)\left(x_i-\frac{1}{\beta_i}\right)=(1-\alpha)\left(x_j-\frac{1}{\beta_j}(1-e^{-\beta_jx_j})\right)(1-e^{-\beta_ix_i})-\alpha\frac{1}{\beta_j}e^{-\beta_jx_j}e^{-\beta_ix_i}\/,$$
    for $(i,j)\in\{(1,2),(2,1)\}$. The exact solution can be determined using numerical optimization methods. For our illustration, we use the Newton-Raphson's multidimensional algorithm.\\ Figure~\ref{Fig-convRM-EI} presents the result obtained for $\beta_1=0.05$ (red) and $\beta_2=0.25$ (blue).
    The exact value of the bivariate expectile $(20.02,3.22)$ is represented by the straight dotted lines of the same colors, respectively. In order to ensure the robustness of the result, we use an average of $100$ outputs of the algorithm. 
    \begin{figure}[H]
    	\center
    	\includegraphics[width=8cm]{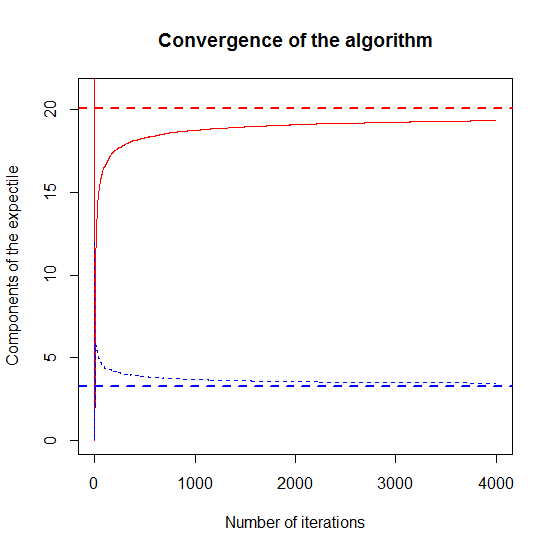} 
    	\includegraphics[width=8cm]{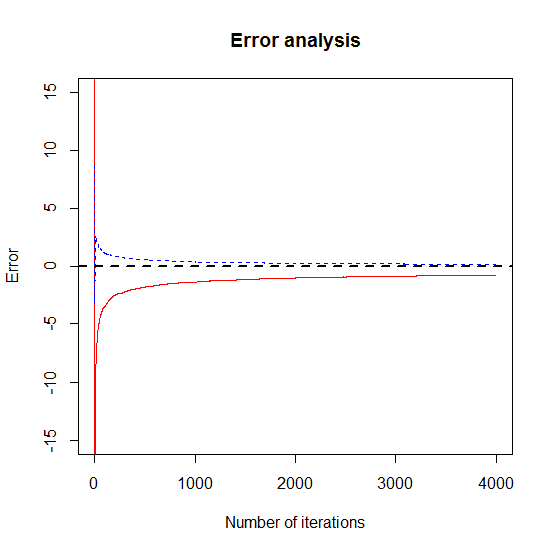}
    	\caption{Convergence of the algorithm: $L_1-$expectile, Exponential independent model (red $\beta_1=0.05$, blue $\beta_2=0.25$).}
    	\label{Fig-convRM-EI}
    \end{figure}
    The graphical analysis of the error's normality is shown in Figure~\ref{Fig-convRM-Err-EI}. The impact of the sequence of steps on the speed of convergence is illustrated in Figure~\ref{Fig-convRM-pas-EI}. The convergence is not very sensitive to the starting point choice. 
    
    The Robbins-Monro's algorithm convergence is studied in Lelong (2007) \cite{lelong2007etude}. Two CLT are presented, one for the choice of $1/n$ as sequence of steps and the other in case of sequences of the form $\gamma_n=\gamma/n^\kappa$, where $\kappa$ is a constant and $1/2<\kappa<1$.\\
    \begin{figure}[H]
    	\center
    	\includegraphics[width=8cm]{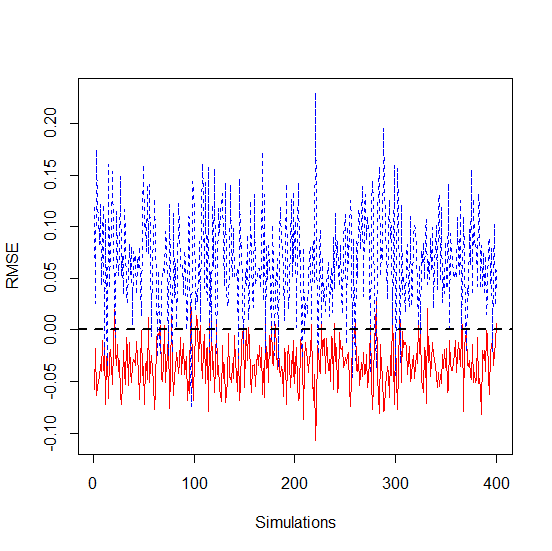} 
    	\includegraphics[width=8cm]{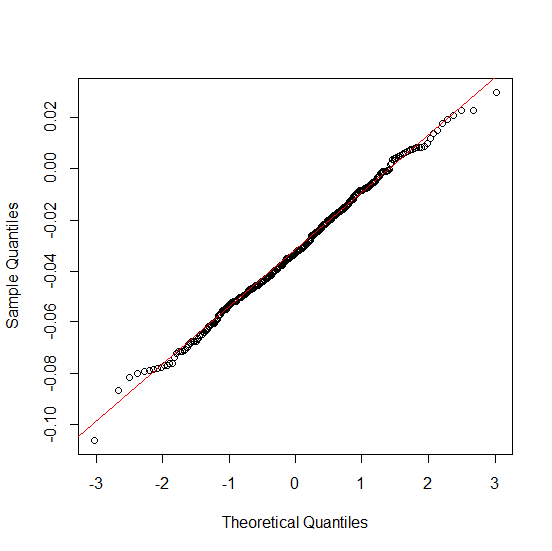}
    	\caption{The algorithm's error, $L_1-$expectile, Exponential independent model.}
    	\label{Fig-convRM-Err-EI}
    \end{figure}
    \begin{figure}[H]
    	\center
    	\includegraphics[width=16cm]{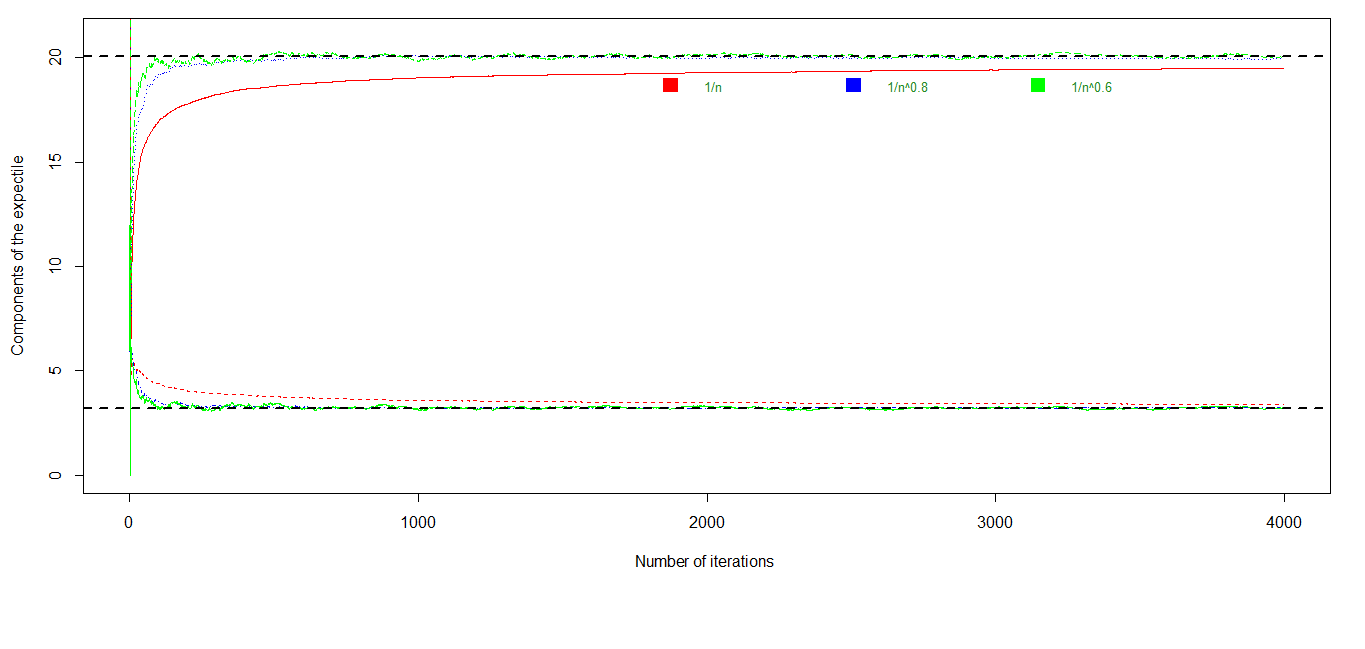} 
    	\caption{Impact of the choice of the sequence of steps on the convergence.}
    	\label{Fig-convRM-pas-EI}
    \end{figure}
    We consider now a bivariate Pareto independent model. The two random variables are of Pareto distribution, $X_i\sim Pa(a,b_i),~i\in\{1,2\}$, such that $a>1$ and $b_i>0$ for all $i\in\{1,2\}$. The system of optimality of $L_1-$expectile is explicit. The expectile $(x_1,x_2)^T$ is the unique solution of the following system 
   $$l^{\alpha}_{X_i}(x_i)=-l^{\alpha}_{X_j,X_i}(x_j,x_i), (i,j)\in\{(1,2),(2,1)\}\/,$$ 
   where,  
   $$l^{\alpha}_{X_i}(x_i)=(2\alpha-1)\frac{b_i}{a-1}\left(\frac{b_i}{b_i+x_i}\right)^{a-1}-(1-\alpha)\left(x_i-\frac{b_i}{a-1}\right)\/,  $$ 
   and $$l^{\alpha}_{X_j,X_i}(x_j,x_i)=\frac{b_j}{a-1}\left(\left(\frac{b_i}{b_i+x_i}\right)^a-(1-\alpha)\right)\left(\frac{b_j}{b_j+x_j}\right)^{a-1}-(1-\alpha)\left(1-\left(\frac{b_i}{b_i+x_i}\right)^a\right)\left(x_j-\frac{b_j}{a-1}\right)\/.$$
   The Newton-Raphson's method is used to get the exact solution.\\
   Figure~\ref{Fig-convRM-PI} is an illustration of the difference in convergence between two different levels $\alpha=0.7$ and $\alpha=0.99$ of the expectile.   
   \begin{figure}[H]
   	\center
   	\includegraphics[width=8cm]{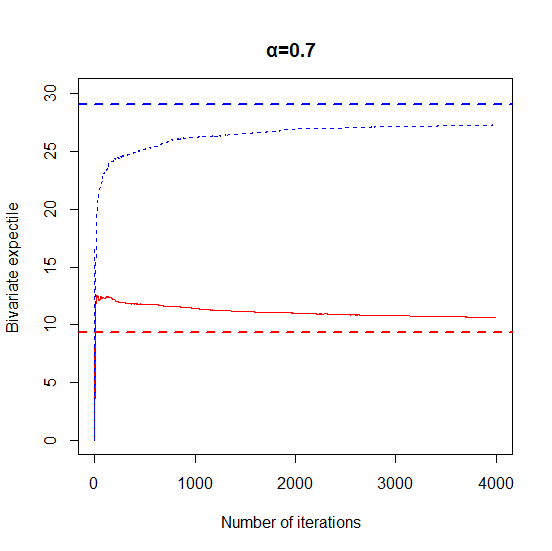} 
   	\includegraphics[width=8cm]{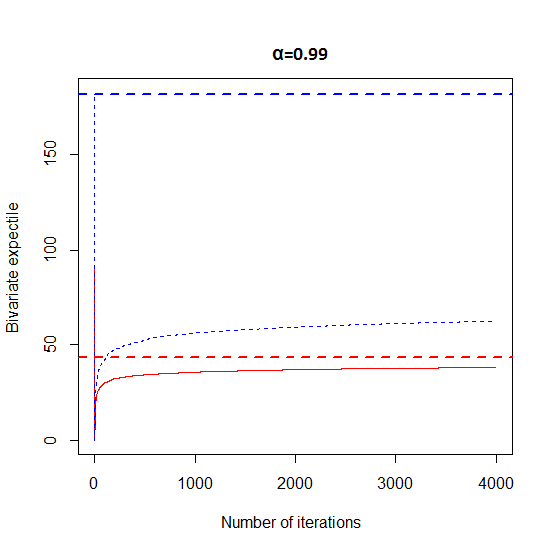}
   	\caption{Convergence of the algorithm, $L_1-$expectile, Pareto independent model ($a=2$, red $b_1=10$, blue $b_2=20$).}
   	\label{Fig-convRM-PI}
   \end{figure}
   Convergence is not very satisfactory for values of $\alpha$ close to $1$. The algorithm is not efficient to estimate the asymptotic expectile. A study of the asymptotic behavior of the expectile seems necessary, particularly in cases where there is no analytical solution. The next section is devoted to the asymptotic expectiles.\\
   
   For the dependence case, we consider a random vector $\mathbf{X}=(X_1,X_2)^T$ of exponential marginals $X_i\sim\mathcal{E}(\beta_i),~i=1,2$ and bivariate FGM copula, as dependence structure, of parameter $\theta\in[-1,1]$. The expression of FGM bivariate copula is  
   $$C^{\mbox{FGM}}_{\theta}(u,v)=uv[1+\theta(1-u)(1-v)],~\forall (u,v)\in[0,1]^2\/.$$
   The FGM copula is presented in details in Nelsen (2007) \cite{nelsen} (Example 3.12., section 3.2.5). We recall that it is only a weak dependence structure that cannot take into account extreme dependencies. The bivariate distribution function is given by
   \begin{align*}
   F_{X_1,X_2}(x_1,x_2)&=C^{\mbox{FGM}}_{\theta}(F_{X_1}(x_1),F_{X_2}(x_2))\\
   &=F_{X_1}(x_1)F_{X_2}(x_2)[1+\theta\bar{F}_{X_1}(x_1)\bar{F}_{X_2}(x_2)]\\
   &=(1-e^{-\beta_1x_1})(1-e^{-\beta_2x_2})+\theta(1-e^{-\beta_1x_1})(1-e^{-\beta_2x_2})e^{-\beta_1x_1}e^{-\beta_2x_2}\/. 
   \end{align*}
   The $L_1$-expectile is the unique solution of the following optimality system 
   \begin{align*}
   (2\alpha-1)\frac{1}{\beta_j}e^{-\beta_jx_j}-(1-\alpha)\left(x_j-\frac{1}{\beta_j}\right)&=(1-\alpha)(1-e^{-\beta_jx_j})(1-\theta e^{-\beta_jx_j})\left(x_i-\frac{1}{\beta_i}(1-e^{-\beta_ix_i})\right)\\&~~+(1-\alpha)\frac{\theta}{2}
  (1-e^{-\beta_jx_j})e^{-\beta_jx_j}\left(2x_i-\frac{1}{\beta_i}(1-e^{-2\beta_ix_i})\right)\\
   &~~-\alpha\frac{1}{\beta_i}e^{-\beta_ix_i}e^{-\beta_jx_j}\left(1+\theta(1-\frac{e^{-\beta_ix_i}}{2})(1-e^{-\beta_jx_j})\right)\/,
   \end{align*}
   for $(i,j)\in\{(1,2),(2,1)\}$.\\ Figure~\ref{Fig-convRM-FGM} presents the result obtained for $\beta_1=0.05$ (red) and $\beta_2=0.25$ (blue), in cases of positive and negative dependence.
   \begin{figure}[H]
   	\center
   	\includegraphics[width=8cm]{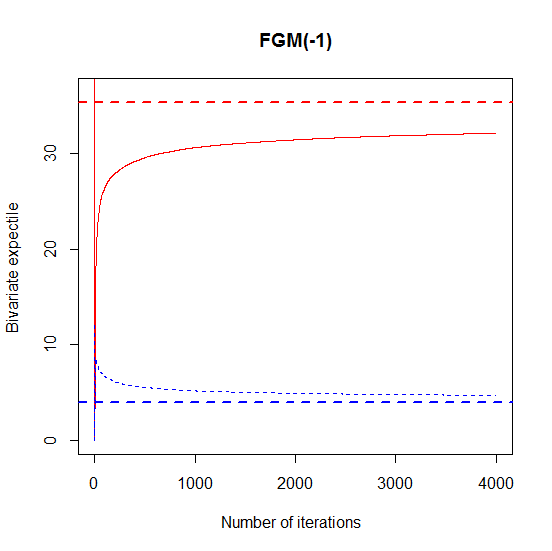} 
   	\includegraphics[width=8cm]{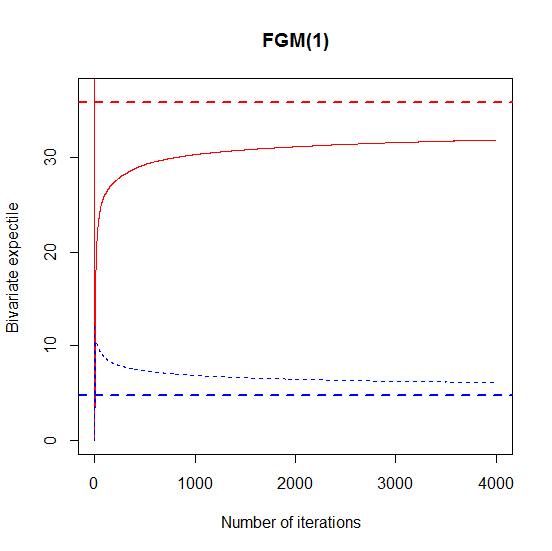}
   	\caption{Convergence of the algorithm: $L_1-$expectile $\alpha=0.85$, FGM model (left $\theta=-1$, right $\theta=1$, red $\beta_1=0.05$, blue $\beta_2=0.25$).}
   	\label{Fig-convRM-FGM}
   \end{figure}
   \section{Asymptotic behavior} 
   Let $(X_1,\ldots,X_d)^T$ be a random vector in $\mathbb{R}^d$. The support of each random variable $X_i$ is denoted $[x^i_I,x^i_F]$, where $x^i_I\in\mathbb{R}\cup\{-\infty\}$ and $x^i_F\in\mathbb{R}\cup\{+\infty\}$. We denote by  $\mathbf{X_F}$ the vector $(x^1_F,\ldots,x^d_F)^T$, and by $\mathbf{X_I}$ the vector $(x^1_I,\ldots,x^d_I)^T$.\\
   We recall the definition of functions    
    $l^\alpha_{X_i,X_j}$ for all $(i,j)\in\{1,\ldots,d\}^2$ as follows
   \begin{equation}\label{FunctionsL}
   l^\alpha_{X_i,X_j}(x_i,x_j)=\alpha\mathbb{E}[(X_i-x_i)_+ 1\!\!1_{\{X_j>x_j\}}]-(1-\alpha)\mathbb{E}[(X_i-x_i)_- 1\!\!1_{\{X_j<x_j\}}]\/,
   \end{equation}
   for all $(x_i,x_j)\in\mathbb{R}^2$. We also denote by $l^\alpha_{X_i}$ the function  $l^\alpha_{X_i}(x_i)=l^\alpha_{X_i,X_i}(x_i,x_i)$.\\
   The System of optimality~\ref{eq1} can be written using these function as
   \begin{equation}\label{SO-fL}
   \sum_{i=1}^{d}\pi_{ki}l^\alpha_{X_i,X_k}(x_i,x_k)=0~~\forall k\in\{1,\ldots,d\}\/.
   \end{equation}
   We focus now on the asymptotic behavior of multivariate expectiles.\\
   The construction matrix $\Sigma$ is supposed constituted of 
   positive coefficients $(\pi_{ij}\geq0,~\forall (i,j)\in\{1,\ldots,d\}^2)$. 
   \begin{Proposition}[Asymptotic Expectiles]\label{asyPro}
   	For any order $2$ random vector $\mathbf{X}=(X_1,\ldots,X_d)^T$ in $\mathbb{R}^d$, where $\mathbb{E}[|X_i|]<+\infty$ for all $i\in\{1,\ldots,d\}$,   
   	$$\underset{\alpha\longrightarrow1}{\lim}\mathbf{e}_\alpha(\mathbf{X})=\mathbf{X_F},~~\mbox{et}~~\underset{\alpha\longrightarrow0}{\lim}\mathbf{e}_\alpha(\mathbf{X})=\mathbf{X_I}
   	\/.$$
   	If in addition, all supports are infinite, then for all \text{$k\in\{1,\ldots,d\}$}
   	$$\underset{x_k\longrightarrow+\infty}{\lim}\alpha(x_1,\ldots,x_d)=1,~~\mbox{et}~~\underset{x_k\longrightarrow-\infty}{\lim}\alpha(x_1,\ldots,x_d)=0\/,$$
   	where \begin{equation}
   	\alpha(x_1,\ldots,x_d)=\frac{\sum_{i=1}^{d}\pi_{ki}\mathbb{E}[(x_i-X_i)_+ 1\!\!1_{\{x_k>X_k\}}]}{\sum_{i=1}^{d}\pi_{ki}\left(\mathbb{E}[(X_i-x_i)_+ 1\!\!1_{\{X_k>x_k\}}]+\mathbb{E}[(x_i-X_i)_+ 1\!\!1_{\{x_k>X_k\}}]\right)}\/.
   	\end{equation} 	
     \end{Proposition}
   \begin{proof}
   	It is sufficient to prove that  
   	$$\underset{\alpha\longrightarrow1}{\lim}\mathbf{e}_\alpha(\mathbf{X})=\mathbf{X_F}\/,$$
   	and we deduce the limit for $\alpha\longrightarrow0$ using the property of symmetry by $\alpha$ 
   	$\mathbf{e}_\alpha(-\mathbf{X})=-\mathbf{e}_{1-\alpha}(\mathbf{X})$.\\
   	For simplicity, we make the proof for $\pi_{ij}=1 \forall (i,j)\in\{1,\ldots,d\}^2$. The generalization is straightforward.\\
   	Taking if necessary a convergent subsequence, we consider that the limit $\underset{\alpha\longrightarrow1}{\lim}\mathbf{e}_\alpha(\mathbf{X})$ exists.\\
   	We define consider
   $$J_\infty:=\{i\in\{1,\ldots,d\} | \underset{\alpha\longrightarrow1}{\lim}\mathbf{e}^i_\alpha(\mathbf{X})=x^i_F\}\/.$$
  Its complementary set is denoted $\bar{J}_\infty$
   $$\bar{J}_\infty:=\{i\in\{1,\ldots,d\} | i\notin J_\infty\}=\{i\in\{1,\ldots,d\} | \underset{\alpha\longrightarrow1}{\lim}\mathbf{e}^i_\alpha(\mathbf{X})<x^i_F\}\/.$$
   We suppose firstly $J_\infty=\emptyset$. Then,  $\forall ~i,j\in\bar{J}_\infty=\{1,\ldots,d\}$,
   $$\underset{\alpha\longrightarrow1}{\lim}l^\alpha_{X_i}(x_i)>0$$ and $$ \underset{\alpha\longrightarrow1}{\lim}l^\alpha_{X_i,X_j}(x_i,x_j)>0\/.$$
   That is absurd, because it is  contradictory with the system of optimality (\ref{SO-fL}). From that, we deduce  $J_\infty\neq\emptyset$.\\
   There exists at least one  $k\in\{1,\ldots,d\}$ such that  \text{$\underset{\alpha\longrightarrow1}{\lim}\mathbf{e}^k_\alpha(\mathbf{X})=x^k_F$}. We have 
   $$\underset{\alpha\longrightarrow1}{\lim}l^\alpha_{X_i,X_k}(\mathbf{e}^i_\alpha(\mathbf{X}),\mathbf{e}^k_\alpha(\mathbf{X}))=0, ~\forall i\in\bar{J}_\infty\/,$$ 
   and 
   $$\underset{\alpha\longrightarrow1}{\lim}\sum_{i\in J_\infty}^{}l^\alpha_{X_i,X_k}(\mathbf{e}^i_\alpha(\mathbf{X}),\mathbf{e}^k_\alpha(\mathbf{X}))=-\underset{\alpha\longrightarrow1}{\lim}\sum_{i\in J_\infty}\left((1-\alpha)\mathbb{E}[(X_i-\mathbf{e}^i_\alpha(\mathbf{X}))_-1\!\!1_{\{X_k<\mathbf{e}^k_\alpha(\mathbf{X})\}}]\right)=0\/,$$
   by (\ref{SO-fL}).\\  
   We deduce for $i=k$
   \begin{equation*}
   \underset{\alpha\longrightarrow1}{\lim}l^\alpha_{X_i}(\mathbf{e}^i_\alpha(\mathbf{X}))=-\underset{\alpha\longrightarrow1}{\lim}\left((1-\alpha)\mathbb{E}[(X_k-\mathbf{e}^i_\alpha(\mathbf{X}))_-]\right)=0, ~~\forall i\in J_\infty\/.
   \end{equation*}
   (\ref{SO-fL})then leads to 
   \begin{equation}\label{LimK}
   \underset{\alpha\longrightarrow1}{\lim}l^\alpha_{X_k}(\mathbf{e}^k_\alpha(\mathbf{X}))=-\underset{\alpha\longrightarrow1}{\lim}\left((1-\alpha)\mathbb{E}[(X_k-\mathbf{e}^k_\alpha(\mathbf{X}))_-]\right)=0, ~~\forall k\in J_\infty\/.
   \end{equation}	
   If we assume that $\bar{J}_\infty\neq\emptyset$, there exists $\ell\in\{1,\ldots,d\}$  such that $\underset{\alpha\longrightarrow1}{\lim}\mathbf{e}^\ell_\alpha(\mathbf{X})<x^\ell_F$.
   In this case,
   $$\underset{\alpha\longrightarrow1}{\lim}l^\alpha_{X_i,X_\ell}(\mathbf{e}^i_\alpha(\mathbf{X}),\mathbf{e}^\ell_\alpha(\mathbf{X}))=\mathbb{E}[(X_i-\underset{\alpha\longrightarrow1}{\lim}\mathbf{e}^i_\alpha(\mathbf{X}))_+1\!\!1_{\{X_\ell>\underset{\alpha\longrightarrow1}{\lim}\mathbf{e}^\ell_\alpha(\mathbf{X})\}}]\in\mathbb{R}^+\backslash\{+\infty\},~\forall i\in\bar{J}_\infty\/,$$
   and using (\ref{LimK}) 
   $$\underset{\alpha\longrightarrow1}{\lim}l^\alpha_{X_i,X_\ell}(\mathbf{e}^i_\alpha(\mathbf{X}),\mathbf{e}^\ell_\alpha(\mathbf{X}))=-\underset{\alpha\longrightarrow1}{\lim}\left((1-\alpha)\mathbb{E}[(X_i-\mathbf{e}^i_\alpha(\mathbf{X}))_-1\!\!1_{\{X_\ell<\underset{\alpha\longrightarrow1}{\lim}\mathbf{e}^\ell_\alpha(\mathbf{X})\}}]\right)=0,~\forall i\in J_\infty\/, $$
   because $$\mathbb{E}[(X_i-\mathbf{e}^i_\alpha(\mathbf{X}))_-1\!\!1_{\{X_\ell<\underset{\alpha\longrightarrow1}{\lim}\mathbf{e}^\ell_\alpha(\mathbf{X})\}}]\leq\mathbb{E}[(X_i-\mathbf{e}^i_\alpha(\mathbf{X}))_-\/,]$$ and $$\underset{\alpha\longrightarrow1}{\lim}\left((1-\alpha)\mathbb{E}[(X_i-\mathbf{e}^i_\alpha(\mathbf{X}))_-]\right)=0\/,$$\\
  for all $i\in J_\infty$.\\ 
   Finally, using the $\ell^{\mbox{th}}$ equation of optimality system (\ref{SO-fL}), we conclude that $$\mathbb{E}[(X_i-\underset{\alpha\longrightarrow1}{\lim}\mathbf{e}^i_\alpha(\mathbf{X}))_+1\!\!1_{\{X_\ell>\underset{\alpha\longrightarrow1}{\lim}\mathbf{e}^\ell_\alpha(\mathbf{X})\}}]=0,~\forall i\in\bar{J}_\infty\/,$$
   and in particularly that 
   $$\mathbb{E}[(X_\ell-\underset{\alpha\longrightarrow1}{\lim}\mathbf{e}^\ell_\alpha(\mathbf{X}))_+]=0,$$
   that is contradictory with the assumption 
  $\underset{\alpha\longrightarrow1}{\lim}\mathbf{e}^\ell_\alpha(\mathbf{X})<x^\ell_F$. We deduce therefore that  $\bar{J}_\infty=\emptyset$.\\	 
   	The second part of Proposition~\ref{asyPro} is straightforward in the case of infinite supports.	 	
  \end{proof}
   The multivariate expectile tends to the vector of the marginal endpoints when $\alpha\rightarrow 1$.  The asymptotic behavior models the situation of extreme risk, hence the practical importance of its study, especially in insurance.  
   
  \section*{Conclusion}
 
 In this paper, we have presented different approaches to construct some multivariate risk measures. The starting point of these methods was the elicitability property.  In a second time, we have chosen a specific construction using matrices to study its coherence properties. Multivariate expectiles are obtained using positive semi-definite matrices with positive coefficients, they allow dependence modeling and take into account the nature of marginal distributions. We also proposed a stochastic approximation method for this family of measures, based on the Robbins-Monro's algorithm. For asymptotic levels of the threshold, the approximation does not provide relevant information on the behavior of expectile vector. A natural perspective of this work is a theoretical analysis to understand the impact of dependence on the asymptotic behavior. 
\bibliographystyle{plain}
\bibliography{biblMultivExp} 
\end{document}